\documentclass{easychair}

\usepackage{amsfonts}
\usepackage{amsthm}
\usepackage[english]{babel}
\usepackage{amsmath}
\usepackage{amssymb}

\usepackage{graphicx}
\graphicspath{{./images}}
\usepackage{mdframed,wrapfig,subcaption}
\usepackage{url}
\usepackage{marginnote}
\usepackage[numbers]{natbib}
\usepackage{stmaryrd}
\usepackage{adjustbox}
\usepackage{local-macros}
\usepackage[title]{appendix}
\usepackage{todonotes}
\usepackage{bm}
\usepackage{syntax}
\usepackage{proof}
\usepackage{soul}
\usepackage{pgfgantt}
\usepackage{amstext}
\usepackage{empheq}
\usepackage{mathtools}
\usepackage{bussproofs}
\usepackage{comment}
\usepackage{tcolorbox}
\usepackage{paralist}

\theoremstyle{definition}
\newtheorem{example}{Example}[section]
\newtheorem{definition}{Definition}[section]

\newtheorem{theorem}{Theorem}[section]
\newtheorem*{theorem*}{Theorem}
\newtheorem{lemma}{Lemma}[section]
\newtheorem*{lemma*}{Lemma}

\newtcolorbox{authorComment}[1]{colback=#1}

\begin{document}
	\title{Logic of Differentiable Logics: Towards a Uniform Semantics of DL}
	\author{Natalia \'{S}lusarz\inst{1} \and Ekaterina Komendantskaya\inst{1} \and Matthew L. Daggitt\inst{1} \and Robert Stewart\inst{1} \and Kathrin Stark\inst{1}}
	\institute{
		Heriot-Watt University, Edinburgh, United Kingdom \\
		nds1@hw.ac.uk
	}
\authorrunning{\'{S}lusarz, Komendantskaya, Daggitt, Stewart and Stark}
\titlerunning{LDL: the Logic of Differentiable Logics}
	\maketitle
	
\begin{abstract}

Differentiable logics (DL) have recently been proposed as a method of training neural networks to satisfy logical specifications. A DL consists of a syntax in which specifications are stated and an interpretation function that translates expressions in the syntax into loss functions.
These loss functions can then be used during training with standard gradient descent algorithms. 
The variety of existing DLs and the differing levels of formality with which they are treated makes a systematic comparative study of their properties and implementations difficult. 
This paper remedies this problem by suggesting a meta-language for defining DLs that we call the Logic of Differentiable Logics, or LDL. Syntactically, it generalises the syntax of existing DLs to FOL, and for the first time introduces the formalism for reasoning about vectors and learners. 
Semantically, it introduces a general interpretation function that can be instantiated to define loss functions arising from different existing DLs.    
We use LDL to establish several theoretical properties of existing DLs and to conduct their empirical study in neural network verification. 

\end{abstract}
\section{Introduction}
\label{sec:introduction}

Logics for reasoning with uncertainty and probability have long been known and used in programming: e.g. fuzzy logic~\cite{zadeh1965fuzzy}, probabilistic logic~\cite{nilsson1986probabilistic} and variants of thereof in logic programming domain~\cite{de2015probabilistic, lukasiewicz1998probabilistic, ng1992probabilistic, riguzzi2022foundations}. Recently, rising awareness of the problems related to \emph{machine learning verification} opened a novel application area for those ideas.
\emph{Differentiable Logics} (DLs) is a family of methods that applies key insights from fuzzy logic and probability theory to enhance this domain with property-driven learning~\cite{giunchiglia2022deep}. 

As a motivating example, consider the problem of verification of neural networks. A neural network is a function
$f: \Real^n \rightarrow \Real^m$ parametrised by a set of weights $\ws$. A training dataset~$\dataset$ is a set of pairs~$(\x,\y)$ consisting of an input $\x \in \Real^n$ and the desired output $\y \in \Real^m$. It is assumed that the outputs $\y$ are generated from $\x$ by some function~$\hypothesis : \Real^n \rightarrow \Real^m$ and that $\x$ is drawn from some probability distribution over $\Real^n$.
The goal of training is to use the dataset~$\dataset$ to find weights~$\ws$ such that~$f$ approximates~$\hypothesis$ well over input regions with high probability density.
 The standard approach is to use a \emph{loss function} $\lossfn$, that given a pair $(\x, \y)$ calculates how much $f(\x)$ differs from the desired output $\y$. 
 Gradients of $\lossfn$ with respect to the network's weights can then be used to update the weights during training.  

However in addition to the dataset $\dataset$, in certain problem domains we may have additional information about $\hypothesis$ in the form of a mathematical property $\property$ that we know $\hypothesis$ must satisfy. 
Given that $\hypothesis$ satisfies property $\property$, we would like to ensure that $f$ also satisfies $\property$. 
A common example of such a property is that $\hypothesis$ should be \emph{robust}, i.e.~small perturbations to the input only result in small perturbations to the output. For example, in image classification tasks changing a single pixel should not significantly alter what the image is classified as~\cite{CKDKKAE22,szegedy2013intriguing}.
\begin{definition}[$\epsilon$-$\delta$-robustness]
\label{def:robustness}
Given constants $\epsilon, \delta \in \Real$, a function $f$ is \emph{$\epsilon$-$\delta$-robust} around a point $ \hat{\x} \in \Real^n$ if:
\begin{equation}
\forall \x \in \Real^n : ||\x - \hat{\x}|| \leq \epsilon \implies || f(\x) - f(\hat{\x})|| \leq \delta
\end{equation}
\end{definition}
\noindent 
\citet{CKDKKAE22} formulate several possible notions of robustness. Throughout this paper, we will use the term ``robustness property" to refer to the above definition.
The problem of verifying the robustness of neural networks has received significant attention from the verification community~\cite{SinghGPV19,WangZXLJHK21}, and it is known to be difficult both theoretically~\cite{KaBaDiJuKo17Reluplex} and  in practice~\cite{BaLiJo21}. 
 However, even leaving aside the challenges of undecidability of non-linear real arithmetic~\cite{KKKAA20}, and scalability~\cite{WangZXLJHK21}, the biggest obstacle to verification is that the majority of neural networks do not succeed in learning $\property$ from the training dataset $\dataset$~\cite{giunchiglia2022deep, wang2018efficient}.
 
The concept of a \emph{differentiable logic} (DL) was introduced to address this challenge by verification-motivated training. This idea is sometimes referred to as continuous verification~\cite{KKK20,CKDKKAE22}, referring to the loop between training and verification.
The key idea in differentiable logic is to use $\property$ to generate an alternative \emph{logical loss function} $\lossfn_\property$, that calculates a penalty depending on how much $f$ deviates from satisfying $\property$. When combined with the standard data-driven loss function $\mathcal{L}$, the network is trained to both fit the data and satisfy $\property$. A DL therefore has two components: a suitable language for expressing the properties and an interpretation function that can translate expressions in the language into a suitable loss function.
 
Although the idea sounds simple, developing good principles of DL design has been surprisingly challenging. 
The machine-learning community has proposed several DLs such as DL2 for supervised learning~\cite{fischer2019dl}, and a signal temporal logic based-DL for reinforcement learning (STL)~\cite{varnai}. 
However, both approaches had shortcomings from the perspective of formal logic: 
the former fell short of introducing quantifiers as part of the language (and thus modelled robustness semi-formally in all experiments), 
and the latter only covered propositional fragment (which is not sufficient to reasoning about robustness).

Solutions were offered from the perspective of formal logic. In ~\cite{KriekenAH20,VANKRIEKEN2022103602} it was shown that one can use various propositional fuzzy logics to create loss functions. However, these fuzzy DLs did not stretch to cover the features of the DLs that came from machine learning community, and in particular could not stretch to express the above robustness property (Definition~\ref{def:robustness}), which needs not just quantification over infinite domains (quantifiers for finite domains were given in~\cite{KriekenAH20}), but also a formalism to express properties concerning vectors and functions over vectors. 

The first problem is thus: \emph{There does not exist a DL that formally covers a sufficient fragment of first-order logic to express key properties in machine learning verification, such as robustness}.    

The second problem has to do with formalisation of different DLs. 
\emph{In many of the existing DL approaches  syntax, semantics and pragmatics are not well-separated, which inhibits their formal analysis.} We have already given an example of DL2 treating quantifiers empirically outside of the language.
But the problem runs deeper than modelling quantifiers. To illustrate, let us take a fragment of syntax on which all DLs are supposed to agree. 
It will give us a toy propositional language $$a := p\ |\ a \wedge a$$
 
They assume that each propositional variable $p$ is interpreted in a domain $\domainSymbol \subseteq \Real$. The domains vary vastly across the DLs (from fuzzy set $[0,1]$ to $(\infty, \infty)$ in STL) and the choice of a domain can have important ramifications for both the syntax and the semantics.  For example DL2's domain $[0, \infty]$ severely restricts the translation of negation compared to other DLs. 
Conjunction is interpreted by $+$ in DL2~\cite{fischer2019dl}, by $min$, $\times$ and other more complex operations in different fuzzy DLs in \citet{VANKRIEKEN2022103602}.  
 In STL~\cite{varnai}, in order to make loss functions satisfy a \emph{shadow-lifting property}, 
  the authors propose a complex formula computing conjunction, that includes natural exponents alongside other operations.
   But this forces them to redefine the syntax for conjunction and thus obtain a different language $$a := p | \bigwedge_M (a_1,\ldots,a_M)$$  where $\bigwedge_M$ denotes a (non-associative!) conjunction for $M$ elements.

As consequence of the above two problems, the third problem is \emph{lack of unified, general syntax and semantics able to express multiple different DLs and modular on the choice of DL, that would make it possible to choose one best suited for concrete task or design new DLs in an easy way.} 

In this paper, we propose a solution to all of these problems. The solution comes in a form of a meta-DL, which we call  \emph{a logic of differentiable logics (LDL)}. On the syntactic side (Section~\ref{sec:syntax}),
it is a typed first-order language with negation and universal and existential quantification that can express properties of functions and vectors.
  
 On the semantic side (Section~\ref{sec:semantics}), interpretation is defined to be parametric on the choice of the interpretation domain $\domainSymbol$ or a particular choice of logical connectives. 
  This parametric nature of interpretation simplifies both the theoretical study and implementations that compare different DLs. Moreover, the language has an implicit formal treatment of neural networks via a special kind of context -- a solution that we found necessary in achieving a sufficient level of generality in the semantics. For the first time the semantics formally introduces the notion of a probability distribution that corresponds to the data from which the data is assumed to be drawn. 
  We demonstrate the power of this approach by using LDL to 
  prove soundness of various DLs in Sections~\ref{subsec:soundness}, \ref{subsec:soundness2}  and systematically compare their geometric properties in Section~\ref{sec:properties}.  
  In Section~\ref{sec:eval} we use LDL 
  to provide a uniform empirical comparison of performance of all mentioned DLs 
  on improving robustness.

\section{Background}

In the previous section, we informally introduced the notions of probability distributions and loss functions. 
We now formally define these. In what follows, quantities such as $\x$ which are written using a bold font refer to vectors and the notation $x_i$ refers to the $i^{th}$ element of $\x$. 

\subsection{Probability Distributions and Expectations}
\label{sec:datasets-and-classifiers}

Following the standard definitions~\cite{Goodfellow-et-al-2016,roussas2003introduction}, we start by considering a random variable $X$ that ranges over some domain $D$.  The probability distribution for $X$ characterises how probable it is for a sample from it to take a given value in the domain $D$.  Depending on whether $D$ is discrete or continuous, the variable is called \emph{discrete or continuous random variable}, respectively. 
Formally, given a continuous random variable $X$ with domain $D$, \emph{the probability distribution function (PDF) $ \distr_X : D \rightarrow [0,1]$} 
is a function that satisfies: $ \int_{D} \distr_{X}(x) dx = 1$.

This definition can be generalised to random vector variables $\X$ over the domain $\D = D_1 \times \ldots \times D_n$ as follows.
Each element of the vector is assumed to be drawn from a random variable $X_i$ over domain $D_i$. The \emph{joint PDF} is a function
$\distr_{X_1, \ldots, X_n} : D_1 \times \ldots \times D_n \rightarrow [0,1]$ that satisfies:
\begin{equation*}
\int_{D_1} \ldots \int_{D_n} \distr_{X_1, \ldots, X_n}(x_1 , \ldots , x_n) dx_n \ldots dx_1 = 1
\end{equation*}
For brevity we will write $\distr_\X$ instead of $\distr_{X_1, \ldots, X_n}$ and $\int_\D$ and  $d\x$ rather than the full integral above.
%
%

%



Given a function $g:  \Real^n \rightarrow \Real$, we can calculate an average value that $g$ takes on $\X$, given a probability distribution $\distr_{\X}$.
Formally $\mathbb{E}[g(\X)]$, the \emph{expected value} for $g: \Real^n \rightarrow \Real $ over the random variable $\X$, is defined as:




\begin{equation}
	\mathbb{E}[g(\X)] = \int_\D \distr_{\X}(\x)  g(\x) d\x.
	\label{eq:vector-integral}
\end{equation}
%
%
Throughout the paper we assume that different random variables are independent. 

\subsection{Loss Functions}
\label{sec:loss-functions}
%


In standard machine learning training, given a neural network $f$, a loss function $\losssymbol: \Real^n \times \Real^m \rightarrow \Real$ computes a penalty proportional to the difference between the output of $f$ on a training input~$x$ and a desired output~$y$. There has been a large number of well performing loss functions proposed, with the most popular in classification tasks being cross-entropy loss~\cite{nar2019crossentropy, wang2020comprehensive}:

	Given a classifier  ${f: \Real^n \rightarrow [0,1]^m}$, the cross-entropy loss $\losssymbol_{ce}$ is defined as 
	\begin{equation}\label{eq:ce}
	\losssymbol_{ce}(\x, \y) = - \sum_{i=1}^{m} \y_i \log(f(\x)_i)
	\end{equation}
	where $\y_i$ is the true probability of $\x$ belonging to class $i$ and $f(\x)_i$ the probability for class $i$ as predicted by $f$ when applied to $\x$.

%
%
However, this notion of a loss function is not expressive enough to capture loss functions generated by DLs. Firstly, the property $\property$ may depend on other parameters apart from a labelled input/output pair. For example, the definition of robustness in Definition~\ref{def:robustness}, depends not only on some input $\hat{x}$ but also on parameters $\epsilon$ and $\delta$. Secondly, properties may relate more than one neural network, for example, in student-teacher scenarios~~\cite{hinton2015distilling}.
Therefore we generalise our notion of a loss function as follows:
\begin{definition}[Generalised loss function]
	\label{eq:loss-function}
	Given a set of neural networks $\hNetCtx$ and parameters $\hVarCtx$, a \emph{loss function} $\losssymbol^{\hNetCtx, \hVarCtx}$ is any function of type $\hNetCtx \times \hVarCtx \rightarrow \Real$.
\end{definition}
\noindent Note that we recover cross-entropy loss from this definition, by setting $\hNetCtx = \Real^n \rightarrow \Real^m$ and $\hVarCtx = {\Real^n \times \Real^m}$, whereas in a teacher-student scenario it might be something like $\hNetCtx = (\Real^n \rightarrow \Real^m) \times (\Real^n \rightarrow \Real^m)$ and $\hVarCtx = {\Real^n \times \Real}$. We distinguish between the neural networks and the other parameters, because during training we will differentiate the loss function with respect to the former but not the latter.

\section{Syntax and Type-system of LDL}
\label{sec:syntax}

Figure \ref{fig:syntax} formally defines the syntax of LDL. The first line of the definition of the set  $\langle expr \rangle$ of expressions includes:
bound variables~$x$, neural network variables~$f$ and constants of types $\RealType$, $\text{Index}$ and $\BoolType$. 
The second line of $\langle expr \rangle$ defines the standard syntax for lambda functions, applications and let bindings to facilitate language modularity. In this aspect the syntax is richer then any of the existing DLs.

The third and fourth lines contain standard operations on Booleans (\AndSymbol{}, \OrSymbol{}, \NotSymbol{}, \ImpliesSymbol{}) and Reals  (\AddSymbol{}, \NegSymbol{}, \MulSymbol{}, \GeqSymbol{}, $==$, ...).
The fifth line contains vector operations: \Seq{\exprClass}{\exprClass} is used to construct vectors, and the lookup operation $\AtSymbol{}$ retrieves the value of a vector at the provided index, i.e. $v\ !\ i$ accesses the value at the $i^{th}$ position in the vector $v$.
The final line contains quantifiers.


LDL is a typed language. In the first line of the block defining the set $ \langle type \rangle$, there is the function type constructor $\FunType{\typeClass_1}{\typeClass_2}$, which represents the type of functions which take inputs of type $\typeClass_1$ and produce outputs of type $\typeClass_2$. Next there are the standard Bool and Real types. Finally, there are: Vec $n$, the type of vectors of length $n$, and Index $n$, the type of indices into vectors of length $n$.
The fact that these two types are parametrised by the size of the Vector will allow the type-system to statically eliminate specifications with out-of-bounds errors. 


\begin{figure}[t]
	\small
	\centering
\begin{subfigure}[t]{0.45\textwidth}
	\begin{grammar}
		<expr> $\ni e$ ::=  
		$\identClass$ | $\netClass$ | \elReal $\: \in \Real$ | \elNat $\: \in \Nat$ | \elBool $\: \in \Bool$
		\alt \App{\exprClass}{\exprClass} | \Lam{\identClass}{\typeClass}{\exprClass} | \Let{\identClass}{\typeClass}{\exprClass}{\exprClass}
		\alt \AndSymbol{} | \OrSymbol{} | \NotSymbol{} | \ImpliesSymbol{} | \AddSymbol{} | \NegSymbol{} | \MulSymbol{}
		\alt \NeqSymbol{} | \LeqSymbol{} | \GeqSymbol{} | \LeSymbol{} | \GeSymbol{} | \EqSymbol{}
		\alt \Seq{\exprClass}{\exprClass} | \AtSymbol{}
		\alt \Forall{\identClass}{\typeClass}{\exprClass} | \Exists{\identClass}{\typeClass}{\exprClass} 
		
%
	\end{grammar}
\end{subfigure}
\hfill
\begin{subfigure}[t]{0.5\textwidth}
	\begin{grammar}
		<type> $\ni \tau$  ::= 
		\FunType{s}{\typeClass} 
		| s
		
		<simple type> $\ni s$ ::= 
		\alt \BoolType
		\alt \RealType
		\alt \VecType{n}   |  \FinType{n}  \hspace{3em} for $n \in \Nat$
	\end{grammar}
\end{subfigure}
\setlength{\belowcaptionskip}{-15pt} 
	\caption{Specification Language of LDL: expressions and types. For readability binary operators will often be written using infix notation, but this is syntactic sugar for the prefix form. 
	}
	\label{fig:syntax}
\end{figure}

To illustrate LDL in use, we now present one possible encoding of the robustness property from Definition~\ref{def:robustness} for a network $f$ of size $784$ ($28 \times 28$ pixel images). 
\begin{example}[Encoding of robustness property in LDL]
\label{ex:wte}
{\small{
\newcommand{\exsize}{784} 
\begin{equation*}
\begin{array}{l}
\robustexpr = 
\quad \letText\ (bounded : \VecType{\exsize} 
\rightarrow \VecType{\exsize} 
\rightarrow \RealType 
\rightarrow \BoolType)\ = \\
	
\qquad \lamText\ (\ve : \VecType{\exsize})\ .
\ \lamText\ (\ue : \VecType{\exsize})\ .
\ \lamText\ (a : \RealType)\ . \\
	 
\quad\qquad \forall\ (i : \FinType{\exsize} )\ .
\ \letText\ (d : \RealType) = \ve\ !\ i - \ue\ !\ i \ \textrm{in}

	  -a \leq d \ \land \ d \leq a\ \\
\quad \textrm{in} \\

\quad\quad \lamText\ (\epsilon: \RealType)\ .
\ \lamText\ (\delta : \RealType)\ .
\ \lamText\ (\hat{\x}: \VecType{\exsize})\ . \\

\quad\qquad \forall\ (\x: \VecType{\exsize})\ .\ ( bounded\ \x\ \hat{\x}\ \epsilon)\ \impl \ (bounded\ (f\ \x)\ (f\ \hat{\x})\ \delta)
	\end{array}
	\end{equation*}}}
\end{example}	
\noindent This example illustrates how LDL has already fulfilled several of our goals defined in Section~\ref{sec:introduction}:
\begin{inparaenum}[a)]
	\item 
	While other DLs work with different fragments of propositional and first-order logics, LDL covers the whole subset of FOL, and in particular quantifiers are first-class constructs in the language.
	\item LDL allows one to express properties of neural networks concisely and at a high level of abstraction. The presence of the Vec type makes explicit that the logic is intended to reason over vectors of reals. 
	\item LDL is strongly typed, which will allow us to define a type-system, and subsequently a rigorous notion of an interpretation function over well-typed terms. In our presentation, the types have to be written explicitly, but it would be simple to remove this constraint by using a standard type-inference algorithm to infer most of them.
\end{inparaenum}

\begin{figure}[t]
	\footnotesize
	\begin{gather*}
	\infer[(networkVar)]{\hcTypeRel{f}{\FunType{\VecType{m}}{\VecType{n}}}}
	{\Xi[f] = (m, n)}
	\semSpace
	\infer[(boundVar)]{\hcTypeRel{\id}{\tau}}{\Delta[\id] = \tau}
	\\
	\\
	\infer[(real)]{\hcTypeRel{\elReal}{\RealType}}{\elReal \in \Real}
	\semSpace
	\infer[(index)]{\hcTypeRel{\elNat}{\FinType{n}}}{\elNat \in \{0,...,n-1\}}
	\semSpace
	\infer[(bool)]{\hcTypeRel{\elBool}{\BoolType}}{\elBool \in \ttype{\BoolType}}
	\\
	\\
	\infer[(app)]{\hcTypeRel{\App{e_1}{e_2}}{\tau_2}}{\hcTypeRel{e_1}{\FunType{\tau_1}{\tau_2}} & \hcTypeRel{e_2}{\tau_1}}
	\\
	\\
	\infer[(lam)]{\hcTypeRel{\Lam{\id}{\tau_1}{e}}{\FunType{\tau_1}{\tau_2}}}{\hTypeRel{\hcNetCtx,\consNew{x\rightarrow \tau}{\hcVarCtx}}{e}{\tau_2}}
	\semSpace
	\infer[(let)]{\hcTypeRel{\Let{\id}{\tau_1}{e_1}{e_2}}{\tau_2}}{\hcTypeRel{e_1}{\tau_1}  & \hTypeRel{\hcNetCtx,\consNew{x\rightarrow \tau}{\hcVarCtx}}{e_2}{\tau_2}}
	\\
	\\
	\infer[(and) (or) (implies)]{\hcTypeRel{\wedge,\vee, \Rightarrow}{\FunType{\BoolType}{\FunType{\BoolType}{\BoolType}}}}{}
	\semSpace
	\infer[(not)]{\hcTypeRel{\text{\NotText{}}}{\FunType{\BoolType}{\BoolType}}}{}
	\\
	\\
	\infer[(add) (mul)]{\hcTypeRel{\text{\AddSymbol{}}, \text{\MulSymbol{}}}{\FunType{\RealType}{\FunType{\RealType}{\RealType}}}}{}
	\semSpace
	\infer[(\bowtie)]{\hcTypeRel{\bowtie}{\FunType{\RealType}{\FunType{\RealType}{\BoolType}}}}{}
	\\
	\\
	\infer[(vec)]{\hcTypeRel{\Seq{e_1}{e_n}}{\VecType{n}}}{\hcTypeRel{e_1}{\RealType} & \ldots & \hcTypeRel{e_n}{\RealType}}
	\semSpace
	\infer[(lookup)]{\hcTypeRel{\:\AtSymbol{}}{\FunType{\VecType{n}}{\FunType{\FinType{n}}{\RealType}}}}{}
	\\
	\\
	\infer[(forall)]{\hcTypeRel{\Forall{\id}{\tau}{e}}{\BoolType}}{\hTypeRel{\hcNetCtx,\consNew{x\rightarrow \tau}{\hcVarCtx}}{e}{\BoolType} & \tau \neq \FunType{\tau_1}{\tau_2}}
	\quad
	\infer[(exists)]{\hcTypeRel{\Exists{\id}{\tau}{e}}{\BoolType}}{\hTypeRel{\hcNetCtx,\consNew{x\rightarrow \tau}{\hcVarCtx}}{e}{\BoolType} & \tau \neq \FunType{\tau_1}{\tau_2}}
	\end{gather*}
	\setlength{\belowcaptionskip}{-15pt} 
	\caption{The typing relation $\haTypeRel$ for LDL expressions; $\bowtie$ stands for comparison operators ==, \NeqSymbol{}, \LeqSymbol{}, \GeqSymbol{}, \LeSymbol{},  \GeSymbol{}. }
	\label{fig:types}
\end{figure}

We now define a typing relation for well-typed expressions in LDL. A \emph{bound variable context},~$\Delta$, is a partial function that assigns each bound variable $x$ currently in scope a type $\tau$. We will use the notation $\Delta[x\rightarrow \tau]$ to represent updating $\Delta$ with a new mapping from $x$ to $\tau$. 
A \emph{network context}, $\Xi$, is a function that assigns each network variable $f$ a pair of natural numbers $(m,n)$ such that $ m $ is the number of inputs to the network and $ n $ is the number of outputs. We will use the notation $\Xi[f\rightarrow (m,n)]$ to represent updating $\Xi$ with a new mapping from $f$ to $(m,n)$.
 The set of \emph{well-typed expressions} is the collection of all $e$ for which there exists contexts $\Xi$ and $\Delta$ and type $\tau$ such that $\haTypeRel$, which is defined inductively in Figure~\ref{fig:types}, holds.

When $\Forall{\id}{\tau}{e}$ is well typed and $\tau = \RealType$ or $\tau = \VecType{n}$, we will say that the quantifier $\forall$ is \emph{infinite}, otherwise the quantifier is \emph{finite}; similarly for $\exists$. 
Note that for simplicity, we assume all bound syntactic variable names for quantifiers are unique. In practice this can be achieved by applying typical binding techniques such as de Brujin indices~\cite{debruijn}.

\section{Loss Function Semantics of LDL}
\label{sec:semantics}

We translate the LDL syntax into loss functions in a manner that is parametric on the choice of the concrete DL.
Firstly, in Section~\ref{subsec:types}, we interpret LDL types as sets of values. 
Secondly, in Section~\ref{subsec:expressions} we interpret expressions, splitting the definitions in two parts: expressions whose interpretation is independent of the choice of DL (e.g.~application, lambda, vectors) and expressions whose interpretation is dependent on the choice of DL (e.g.~logical connectives, comparison operators). In Section~\ref{subsec:quantifiers} we define semantics of quantifiers uniformly for all logics.   

\subsection{Semantics of LDL Types}
\label{subsec:types}

We first define a mapping $\ttype{\cdot}$ from LDL types to the set of values that expressions of that type will be mapped to: 
{\small{\begin{equation*}
\newcommand{\tspace}{\quad\:\:\:\:}
	\ttype{\RealType}{} = \Real{}
	\tspace
	\ttype{\VecType{n}}{} = \ttype{\RealType}^{n} 
	\tspace
	\ttype{\FinType{n}}{} = \{0, \ldots, n-1\}
	\tspace
	\ttype{\FunType{\tau_1}{\tau_2}} = {\ttype{\tau_2}{}}^{\ttype{\tau_1}{}}
\end{equation*}}}
\noindent Note that the type Bool is absent as it is dependent on the choice of DL, given in Figure~\ref{fig:expr}. 

\subsection{Semantics of LDL Expressions}
\label{subsec:expressions}

The next step is to interpret well-typed expressions, $\hcTypeRel{e}{\tau}$, which  
will be dependent on the \emph{semantic context} of the interpretation. An expression's semantic context is comprised of three parts: 

\begin{itemize}
	
	\item A semantic network context $\hNetCtx$ is a function that maps each network variable $f \in \Xi $, such that $\Xi[f] = (m,n)$, to a function $\mathtt{f} : \Real^m \rightarrow \Real^n$, the actual (external) implementation of the neural network. We refer to the set of such functions as $\ttype{\Xi}$.
	
	\item Let $q(e)$ be the set of infinitely quantified syntactic variables within expression $e$. A semantic quantifier context $\hRandCtx$ is a function that maps each  variable $x$ in $q(e)$ to a random variable $X$, from which values for $x$ are sampled from (by discussion of Section~\ref{sec:datasets-and-classifiers} this also extends to cover random vector variables). We refer to the set of such functions as~$\ttype{q(e)}$.
	
	\item A semantic bound context $\hVarCtx$ is a partial function that assigns each bound variable $x \in \Delta$ a semantic value in~$\ttype{\Delta(x)}$. We refer to the set of such functions as $\ttype{\Delta}$. 

\end{itemize}
\begin{example}[Semantic context]
Consider the LDL expression in Example~\ref{ex:wte}: $\robustexpr$ contains a single network variable~$f$, and therefore the network context will be of the form: $\hNetCtx = [ f \mapsto (\mathtt{f}: \Real^{784} \rightarrow \Real^{10}) ]$ where $\mathtt{f}$ is the actual neural network implementation (e.g.~the Tensorflow/PyTorch object). The expression contains a single infinitely quantified variable $\x$, and therefore the quantifier context $\hRandCtx$ will be of the form $[\x \mapsto \X]$, where $\X$ is random vector  variable of size $784$, and represents the underlying distribution that the input images are being drawn from. At the top-level of the expression, no bound variables in scope and therefore $\hVarCtx$ is empty. However, when interpreting the subexpression $bounded\ \x\ \hat{\x}\ \epsilon$,the bound context $\hVarCtx$ will be of the form:
$[bounded \mapsto v_1, \: \epsilon \mapsto v_2,\:  
 \: \hat{\x} \mapsto v_3, \: \x \mapsto v_4]$
where $v_1$ is the interpretation of the let bound expression \emph{bounded}, $v_2 \in \Real$ is the concrete value of $\epsilon$, and $v_3, v_4 \in \Real^{784}$ are the value passed in for $\hat{\x}$ and the current value of the quantified variable $\x$ respectively. 
\end{example}

In general the interpretation of an expression will also be dependent on some differentiable logic $\logic$. In this paper, $\logic$ will stand for either of: $DL2$~\cite{fischer2019dl}, $STL$~\cite{varnai}, or Fuzzy DLs ($G$, $L$, $Y$, $p$)~\cite{VANKRIEKEN2022103602}. 
In this section, we present the formalisation for the first three ($DL2$, $STL$, $G$) while the remaining fuzzy logic variants ($L$, $Y$, $p$) can be found in Appendix \ref{appendix:additional-dls}.

Therefore, in general we will use the notation $\texpr{e}{\hNetCtx}{\hRandCtx}{\hVarCtx}{\logic}$ to represent the mapping of LDL expression $e$ to the corresponding value in the loss function semantics using logic $\logic$ in the semantic context $(\hNetCtx, \hRandCtx, \hVarCtx)$. Figure \ref{fig:generic-semantics} shows the definition of $\texpr{\cdot}{\hNetCtx}{\hRandCtx}{\hVarCtx}{\logic}$ for the LDL expressions whose semantics are independent of the choice of $\logic$. 
Note that, with the exception of network variables, the semantics is standard and could belong to any functional language.

Figure~\ref{fig:expr} interprets the expressions that depend on the choice of DL, as follows. 

\begin{figure}[t]
	\centering
	\footnotesize
	\begin{gather*}
	\texpr{\id}{\hNetCtx}{\hRandCtx}{\hVarCtx}{\logic} = \hVarCtx[\id]
	\semSpace
	\texpr{f}{\hNetCtx}{\hRandCtx}{\hVarCtx}{\logic} = \hNetCtx[f]
	\semSpace
	\texpr{\elReal}{\hNetCtx}{\hRandCtx}{\hVarCtx}{\logic} = \elReal
	\semSpace 
	\texpr{\elNat}{\hNetCtx}{\hRandCtx}{\hVarCtx}{\logic} = \elNat
	\\
	\\
	\texpr{\Lam{\id}{\tau}{e}}{\hNetCtx}{\hRandCtx}{\hVarCtx}{\logic} = \lambda (y : \ttype{\tau}) . \texpr{e}{\hNetCtx}{\hRandCtx}{\hVarCtx[x\rightarrow y]}{\logic}
	\semSpace
	\texpr{\App{e_1}{e_2}}{\hNetCtx}{\hRandCtx}{\hVarCtx}{\logic} = \texpr{e_1}{\hNetCtx}{\hRandCtx}{\hVarCtx}{\logic}(\texpr{e_2}{\hNetCtx}{\hRandCtx}{\hVarCtx}{\logic}) 
	\\
	\\
	\texpr{\Let{\id}{\tau}{e_1}{e_2}}{\hNetCtx}{\hRandCtx}{\hVarCtx}{\logic} = \texpr{e_2}{\hNetCtx}{\hRandCtx}{\hVarCtx[x \rightarrow \texpr{e_1}{\hNetCtx}{\hRandCtx}{\hVarCtx}{\logic}]}{\logic}
	\qquad\:
	\texpr{[e_1, ..., e_n]}{\hNetCtx}{\hRandCtx}{\hVarCtx}{\logic} = <{ \texpr{e_1}{\hNetCtx}{\hRandCtx}{\hVarCtx}{\logic}, \ldots, \texpr{e_n}{\hNetCtx}{\hRandCtx}{\hVarCtx}{\logic}}>
	\\
	\\
	\texpr{!}{\hNetCtx}{\hRandCtx}{\hVarCtx}{\logic} = \lam{(\val_1: \ttype{\VecType{n}}), (\val_2:\ttype{\FinType{n}})}{{\val_1}_{\val_2}}
    \\
    \\
	\texpr{+}{\hNetCtx}{\hRandCtx}{\hVarCtx}{\logic} = \lam{(\val_1, \val_2: \ttype{\RealType})}{\val_1 + \val_2} \semSpace
	\texpr{\times}{\hNetCtx}{\hRandCtx}{\hVarCtx}{\logic} = \lam{(\val_1, \val_2: \ttype{\RealType})}{\val_1 \times \val_2}
	\end{gather*}
	\caption{Semantics of LDL expressions that are independent of the choice of DL.}
	\label{fig:generic-semantics}
	\vspace*{-1.5em}
\end{figure}

\begin{figure}[t]
	\footnotesize
	\renewcommand{\arraystretch}{1.3}
	\begin{center}
	\begin{tabular}{|c|c|c|c|}
	\hline 
	Syntax & 
	DL2 interpretation & 
	Gödel interpretation &
	STL interpretation \\ 
	\hline 
	$\tempty{\BoolType}_{L}$ & 
	$[-\infty, 0]$ &
	$ [0, 1] $ & 
	$ [-\infty, \infty] $
	\\ \hline
	$\tempty{\top}_{L}$ & 
	0 & 
	1 &
	\coloured{$\infty$}
	\\ \hline
	$\tempty{\bot}_{L}$ & 
	\coloured{$- \infty$} &
	$0$ &
	\coloured{$- \infty$}
	\\ \hline
	$\tempty{\neg}_{L}$ & 
	- &
	$\lam{\val}1-\val$ &
	$\lam{\val}- \val$
	\\ \hline 
	$\tempty{\wedge}_{L}$ & 
	$ \lam{\val_1,\val_2}{\val_1 + \val_2} $ &
	$ \lam{\val_1,\val_2} \min (\val_1, \val_2)$ & $and_{S}$
	\\ \hline 
	$\tempty{\vee}_{L}$& $ \lam{\val_1,\val_2}{-\val_1 \times \val_2} $ &
	$ \lam{\val_1,\val_2} \max (\val_1, \val_2)$ & 
	$or_{S}$
	\\ \hline 
	$\tempty{\impl}_{L}$ & 
	- & 
	$ \lam{\val_1,\val_2} \max (1 - \val_1, \val_2)$ &
	-
	\\ \hline
	$ \tempty{==}_{L} $ &
	$ \lam{\val_1, \val_2}{- | \val_1 - \val_2 | } $ &
	\coloured{$\lambda \val_1, \val_2. 1 -  \tanh |\val_1-\val_2|$} &
	\coloured{$\lam{\val_1, \val_2}{-|\val_1 - \val_2|}$}
	\\ \hline
	{$ \tempty{\leq}_{L} $} &
	$\lam{\val_1, \val_2}{- \max(\val_1 - \val_2, 0)}$ &
	\coloured{$\lam{\val_1, \val_2}{1-\max(\tanh |\val_1-\val_2|, 0)}$} &
	\coloured{$\lam{\val_1, \val_2}{\val_2 - \val_1}$}
	\\ \hline
	\end{tabular}

\vspace{1em}

\begin{tabular}{lcc}
		$
		and_S  = \lam{\val_1,\ldots,\val_M}
		\begin{cases}
		\dfrac{\sum_i \val_{\min} e^{\tilde{\val_i}} e^{\nu \tilde{\val_i}}}{\sum_i e^{\nu \tilde{\val_i}}} & \text{if}\ \val_{\min} < 0 \\
		\dfrac{\sum_i \val_i e^{-\nu \tilde{\val_i}}}{\sum_i e^{-\nu \tilde{\val_i}}} & \text{if}\ \val_{\min} > 0 \\
		0 & \text{if}\ \val_{\min} = 0 \\
		\end{cases}
		$ 
		where
		\begin{tabular}{c}
		$\nu \in \Real^+$ (constant) \\
		$
		\val_{\min} = \min (\val_1, \ldots, \val_M)
		$\\
		$
		\tilde{\val_i} = \dfrac{\val_i - \val_{\min}}{\val_{\min}} 
		$
		
	\end{tabular} \\\\
		$or_S = $ analogous to $and_S$ & &
	\end{tabular}
	\end{center}
	\caption{Semantics of LDL expressions dependent on the choice of DL. \coloured{Colour} denotes parts of semantics added for LDL and not defined originally. In the implication row ``-'' denotes that interpretation of implication was not provided separately and is defined with negation and disjunction. 
	While $\top$ is part of the syntax in the STL interpretation the semantic interpretation was not provided. We take $\infty$ to stand for $lim_{n \to \infty} n$. }
	\label{fig:expr}
	\vspace{-1em}
\end{figure}

\textbf{Booleans.}
When giving the semantics of LDL types in Section~\ref{subsec:types}, we intentionally omitted the interpretation of $\BoolType{}$, as it is dependent on the DL. In DL2, $\BoolType{}$ was originally mapped to $[0, \infty]$, where $0$ is interpreted as true, and other values were corresponding to the degree of falsity. Thus, $\top$ is interpreted as $0$, but $\bot$ did not exist in DL2. In Fuzzy DLs with the domain $[0,1]$, 
$1$ stood for absolute truth, and other values -- for degrees of (partial) truth; finally, in STL interval $[-\infty, \infty]$, all values but $0$ stood for degrees of falsity and truth. 
Note that we swap the domain of DL2 for $[\infty, 0]$ in order to fit with the ordering of truth values in other logics. 
We will now see how these choices determine interpretation for predicates and connectives.






\textbf{Predicates.}
The predicates are given by comparisons in LDL. Here we interpret just $\leq$ and $==$, the remaining comparisons are given in Appendix~\ref{app:comparisons}. 
%
Originally, only DL2~\cite{fischer2019dl} had these comparisons. 
Inequality of two terms $a_1$ and $a_2$ can be interpreted there by a measure of how different they are. We slightly modify the DL2 translation in order to adapt it in a similar way to other DLs,
%
mapping the difference between $\val_1$ and $\val_2$ inside of the chosen domain. 

\textbf{Logical Connectives.}
The interpretation of $\wedge$ and $\vee$ follows the definitions given in the literature.
As Figure~\ref{fig:expr} makes it clear, the interpretation of negation for DL2 is not defined. By design, DL2 negation is pushed inwards to the level of comparisons between terms, see Appendix~\ref{ap:neg}.  STL translation did not have a defined interpretation for implication.



\subsection{Semantics of LDL Quantifiers}
\label{subsec:quantifiers}

So far, we have given an interpretation for quantifier-free  formulae. However, none of works studied so far have included infinite quantifiers 
as first class constructs in the DL syntax. In this section, we therefore propose novel semantics for $\texpr{\forall x: \tau .\ e}{\hNetCtx}{\hRandCtx}{\hVarCtx}{\logic}$ and $\texpr{\exists x: \tau .\ e}{\hNetCtx}{\hRandCtx}{\hVarCtx}{\logic}$.

Finite quantifiers were introduced in~\cite{KriekenAH20} via finitely composed conjunction and disjunction. We extend this idea to other DLs. Given an expression $\forall x: \tau .\ e$ with a finite quantifier over the variable $x$,
 and given the interpretation $\ttype{\tau}{} = \{d_1, \ldots, d_n\}$,  
%
$\texpr{\forall x: \tau .\ e }{\hNetCtx}{\hRandCtx}{\hVarCtx}{L} = \texpr{\ e [x/d_1] \land \ldots \land  e [x/ d_n]}{\hNetCtx}{\hRandCtx}{\hVarCtx}{L}$; analogously for $\exists$ and $\lor$.
Note that we have only two finite types, $ \FinType{n} $ or $\BoolType$, and for the latter we take $\ttype{\BoolType}{} = \{\texpr{\top}{\hNetCtx}{\hRandCtx}{\hVarCtx}{L}, \texpr{\bot}{\hNetCtx}{\hRandCtx}{\hVarCtx}{L}\}$ to interpret the quantifiers. 

To proceed with infinite quantifiers, recall that context $Q$ gives us a probability distribution for every first-order variable.
This gives us a way to use 
the definition of an expectation for a function from Section~\ref{sec:datasets-and-classifiers} to interpret quantifiers. Recall that, for $g: \Real^n \rightarrow \Real $ we had:
$$\mathbb{E}[g(\X)] = \int_{- \infty}^{\infty} \distr_{\X}(\x)  g(\x) d \x.$$

\citet{fischer2019dl} were the first to interpret universally quantified formulae via expectation maximisation methods, in which case the optimised parameters were neural networks weights, and the quantified properties concerned robustness of the given neural network. 
Our goal is to propose a unifying approach to both universal and existential quantification that will fit with all DLs that we study here,
and will not make any restricting assumptions about the universal properties that the language can express.
For example, we defined LDL to admit expressions that may or may not refer to neural networks (or weights), and may express properties more general than robustness.  


 With this in mind, we introduce the following notation. 
 For a function $g: \Real^n \rightarrow \Real $, we say that $\x_{\min}$ is the \emph{global minimum} (resp. \emph{global maximum}) if 
$g(\x_{\min}) \leq g(\y)$ (resp. $g(\x_{\max}) \geq g(\y)$) for any $\y$ on which $g$ is defined.
We define a $\gamma$-ball around a point $\x$ as follows:
$\mathbb{B}_{\x}^{\gamma} = \{\y \ \mid \  || \x - \y || \leq \gamma  \}$.
We call the expectation  
$$\mathbb{E}_{\min}[g(\X)] = \lim_{\gamma \rightarrow 0} \int_{\x \in \mathbb{B}_{\x_{\min}}^{\gamma}} \distr_{\X}(\x)  g(\x) d \x$$
\emph{minimised expected value for $g$ (over the random variable $\X$)}. Taking $\x \in \mathbb{B}_{\x_{\max}}^{\gamma}$ in the above formula will give the definition of $\mathbb{E}_{\max}[g(\X)]$, the \emph{maximised expected value for $g$ (over the random variable $\X$)}.  

The key insight when applying this is that given a logic $\logic$ and the context $(\semCtx)$, the loss for the body $e$ of the quantified expressions $\forall x: \tau .\ e$ or $\exists x: \tau .\ e$ can be calculated for any particular given semantic value for the quantified variable $x$. Therefore we can construct a function $\lambda y.\texpr{e}{\hNetCtx}{\hRandCtx} {\hVarCtx[x\rightarrow y]}{\logic}$ of type  
$\ttype{\tau} \rightarrow \Real$ that takes in the value and interprets the body of the quantifier with respect to it. 
This gives us interpretation of universal and existential quantifiers as minimised (or maximised) expected values for the interpretation of their body:
\begin{align*}
\texpr{\forall x: \tau .\ e }{\hNetCtx}{\hRandCtx}{\hVarCtx}{L} = \mathbb{E}_{\min}[(\lambda y.\texpr{e}{\hNetCtx}{\hRandCtx}{\hVarCtx[x\rightarrow y]}{L})(Q[x])]
\\
\texpr{\exists x: \tau .\ e }{\hNetCtx}{\hRandCtx}{\hVarCtx}{L} = \mathbb{E}_{\max}[(\lambda y.\texpr{e}{\hNetCtx}{\hRandCtx}{\hVarCtx[x\rightarrow y]}{L})(Q[x])] 
\end{align*}
The formulae above refer to the minimised (resp. maximised) expectations for the function $\lambda y.\texpr{e}{\hNetCtx}{\hRandCtx}{\hVarCtx[x\rightarrow y]}{L}$ over the random variable $Q[x]$. Note that this is the first time we use the semantic quantifier context $Q$ to map a syntactic variable $x$ to a random variable. 
Another notable feature of the resulting interpretation is that the interpretation for quantifiers is parametric on the choice of the logic $\logic$, and that it has the advantage of being compositional which opens a new degree of generality:  it allows for
 arbitrary nesting of quantifiers (in well-formed expressions).  This will greatly simplify the proof of soundness for LDL.

We can now see that DL2's interpretation of quantifiers can be expressed as a special case of this definition. 
In particular, for a quantified formula $\forall x. P(x)$, where $P$ is a robustness property,  ~\citet{fischer2019dl} take $||\x - \hat{\x}|| \leq \epsilon$  and empirically compute the ``worst perturbation" in the $\epsilon$-interval from $\hat{x}$ using a ``PGD adversarial attack" projected within that interval. 
This worst perturbation is the global minimum within the chosen $\epsilon$-ball. And the loss function that optimises the neural network weights minimises the expectation.
This key example motivates our use of expectation terminology in the interpretation of quantifiers. 
An alternative would be to use just the global minimum (or maximum) of a function directly when defining $\texpr{\forall x: \tau .\ e }{\hNetCtx}{\hRandCtx}{\hVarCtx}{L}$ and $\texpr{\exists x: \tau .\ e }{\hNetCtx}{\hRandCtx}{\hVarCtx}{L}$.
But such a choice will not generalise over the DL2 implementation (or any other optimisation method). Another attempt to model quantification over real domains was made by \citet{BadreddineGSS22} within the framework of \emph{``Logic Tensor Networks"}. There, all variables were mapped to finite sequences of real numbers, by-passing explicit use of the notions of random variables and probability distributions over random variables. Intuitively, each given data set has only a finite number of objects, thus giving a finite domain to map to. This solution would not be satisfactory for a DL, that must take into consideration the fact that loss functions are ultimately designed to compute approximations of the unknown probability distribution, from which the given data is sampled (cf. also the discussion given in Introduction).

\begin{example}[Semantics of Quantified Expressions]
	We calculate a loss function for the robustness property in Example~\ref{ex:wte} with respect to DL2 logic in the context $(\semCtx)$ where $\hNetCtx = [ f \mapsto (\mathtt{f}: \Real^{784} \rightarrow \Real^{10}) ]$, $\hRandCtx = [x \mapsto \X]$ and $\hVarCtx$ is empty. Since DL2 has no separate definition of implication we translate the implication in robustness property in a standard manner using negation and disjunction and pushing the negation inwards  to the level of comparisons: 
	\begin{align*}
	\texpr{e^*}{\hNetCtx}{\hRandCtx}{\hVarCtx}{DL2} = \lambda\epsilon.\lambda\delta.\lambda\hat{\x}.\ \mathbb{E}_{\min}[(\lambda \x. -(\tempty{\text{bounded}}\ \x\ \hat{\x}\ \epsilon) \times (\tempty{\text{bounded}}\  f(\x)\ f(\hat{\x})\ \delta))(\X)] \\
\quad \text{where} \ \tempty{\text{bounded}} = \lambda \x . \lambda \y . \lambda v . \sum_{i=0}^{783} -\max (-v-(\x_i - \y_i),0) - \max(v-(\x_i - \y_i),0)
	\end{align*}

\end{example}

The final problem to consider is whether the interpretation of infinite quantifiers via global minima and maxima commutes with the interpretation for $\land$ and $\lor$: in general it does not. Only the G\"{o}del interpretation for these connectives (based on $min$ and $max$ operators) commutes with quantifiers interpreted as minimised or maximised expected values. Section~\ref{sec:properties} considers this problem in detail.

\section{Using LDL to explore properties of DLs}
\label{sec:soundness}


\subsection{Soundness and Completeness of DLs}

In all existing work, soundness and completeness of DLs was proved for propositional fragments of those logics. In the proofs the standard Boolean semantics of propositional classical logic is taken as a decidable procedure for membership of the set of all true formulae $\texpr{\top}{\hNetCtx}{\hRandCtx}{\hVarCtx}{L}$. The main soundness result establishes that if $\texpr{e}{\hNetCtx}{\hRandCtx}{\hVarCtx}{L} \in \texpr{\top}{\hNetCtx}{\hRandCtx}{\hVarCtx}{L}$ then $e$ is \emph{true} in the propositional logic.
Completeness shows that the implication holds in the other direction as well. 
%

However, as LDL is equipped with quantifiers such a procedure is no longer decidable, we may no longer rely on this method.  
For a typed FOL, one could define a Kripke-style semantics~\cite{sorensen2006lectures,LIPTON20181} for characterising the set of all true formulae, or take provable FOL formulae to characterise the set of true FOL formulae. In this section we take the latter approach. 
%
%
Specifically, we take the set of FOL formulae provable in logic \LJ\ by Gentzen~\cite{Gentzen69}. We will say that a DL is \emph{sound}, if the set of formulae that it interprets as $\texpr{\top}{\hNetCtx}{\hRandCtx}{\hVarCtx}{L}$ is a subset of formulae provable in \LJ.

\subsubsection{Type Soundness of Interpretation}

We start by proving soundness of the typing relation in LDL.
The following result, proven by induction on the typing judgement, will be useful in proving soundness of LDL:

\begin{theorem}[Type Soundness of LDL]\label{th:type-sound}
For all differentiable logics $L$ and well typed expressions $\hcTypeRel{e}{\tau}$, then for all $\hNetCtx \in \ttype{\Xi}$, $\hRandCtx \in \ttype{q(e)}$ and $\hVarCtx \in \ttype{\Delta}$
 we have $\tempty{e}^{\hNetCtx,\hRandCtx,\hVarCtx}_L \in \ttype{\tau}$.
\end{theorem}

\begin{proof}
See Appendix~\ref{ap:type-soundness}.
\end{proof}

We will use the fact that all expressions of type $\RealType$ can be interpreted as real numbers in Section~\ref{subsec:soundness}. 

\begin{lemma}[Arithmetic evaluation of real expressions]\label{lem:expr-real}
	If $\validForm{e: \RealType}$ then $\tempty{e}^{\hNetCtx,\hRandCtx,\hVarCtx}_L \in \Real$.
\end{lemma}

\begin{proof}
Obtained as a corollary of Theorem~\ref{th:type-sound} by instantiating $\tau$ with $\RealType$.
\end{proof}


\subsubsection{Sequent Calculus \LJ}

Figure~\ref{fig:terms} gives a formal definition of well-formed FOL formulae in LDL. This definition is a basis for defining sequents in \LJ.
\begin{figure}[!]
	\footnotesize
	\begin{spreadlines}{7pt}
		\begin{empheq}{gather*}
		\def\defaultHypSeparation{\hskip .05in}
		\AxiomC{$\Xi,\Delta \Vdash e : \tau_1 \to \BoolType$}
		\AxiomC{$\Xi,\Delta \Vdash e_1:\tau_1$}
		\BinaryInfC{$\validForm{e_1 \> e_2 }$}
		\bottomAlignProof
		\DisplayProof
		\quad
		\AxiomC{$\validForm{e_1}$}
		\AxiomC{$\validForm{e_2}$}
		\AxiomC{$\Box \in \set{\conj, \disj, \impl}$}
		\TrinaryInfC{$\validForm{e_1 \mathbin{\Box} e_2}$}
		\bottomAlignProof
		\DisplayProof
		\\[7pt]
		\AxiomC{$\validForm[\Xi,\Delta]{e}$}
		\UnaryInfC{$\validForm{\neg e }$}
		\bottomAlignProof
		\DisplayProof
		\quad
		\AxiomC{$\validForm[\Xi,\Delta, x:\tau]{e}$}
		\UnaryInfC{$\validForm{\all{x:\tau} e}$}
		\bottomAlignProof
		\DisplayProof
		\quad
		\AxiomC{$\validForm[\Xi,\Delta, x : \tau]{e}$}
		\UnaryInfC{$\validForm{\exist{x : \tau} e}$}
		\bottomAlignProof
		\DisplayProof
		\end{empheq}
	\end{spreadlines}
	\vspace*{-1em}
	\caption{\footnotesize{(Well-formed) FOL Formulae in LDL.}}
	\label{fig:terms}
\end{figure}

\begin{figure}[bt]
	\footnotesize{
		\begin{spreadlines}{7pt}
			\begin{empheq}{gather*}
			\def\ScoreOverhang{1pt}
			\def\defaultHypSeparation{\hskip .15in}
			\def\labelSpacing{2pt}
			\def\ScoreOverhang{1pt}
			\def\labelSpacing{2pt}
			\AxiomC{$\SequentCLJ{ \tempty{e_1}^{\hNetCtx,\hRandCtx,\hVarCtx}_L \bowtie^* \tempty{e_2}^{\hNetCtx,\hRandCtx,\hVarCtx}_L $}}
			\RightLabel{\LJArithR}
			\UnaryInfC{$\SequentCLJ{e_1 \bowtie e_2}$}
			\bottomAlignProof
			\DisplayProof
			\quad
			\AxiomC{$\SequentCLJ(\Gamma_T, \tempty{e_1}^{\hNetCtx,\hRandCtx,\hVarCtx}_L \bowtie^* \tempty{e_2}^{\hNetCtx,\hRandCtx,\hVarCtx}_L ){e}$}
			\RightLabel{\LJArithL}
			\UnaryInfC{$\SequentCLJ(\Gamma_T,e_1 \bowtie e_2){e}$}
			\bottomAlignProof
			\DisplayProof
			\\
			\AxiomC{$e' \in \ThAssumsEnv$}
			\AxiomC{$e \conv e'$}
			\RightLabel{\LJAxiom}
			\BinaryInfC{$\SequentCLJ{e}$}
			\bottomAlignProof
			\DisplayProof
			\quad
			\AxiomC{}
			\RightLabel{($\bot$)}
			\UnaryInfC{$\SequentCLJ(\ThAssumsEnv, \bot){e}$}
			\bottomAlignProof
			\DisplayProof
		        \quad
			\AxiomC{}
			\RightLabel{($\top$)}
			\UnaryInfC{$\SequentCLJ{\top}$}
			\bottomAlignProof
			\DisplayProof
			\quad 
			\AxiomC{$\SequentCLJ(\ThAssumsEnv, e_1){e}$}
			\RightLabel{\LJNegR}
			\UnaryInfC{$\SequentCLJ(\ThAssumsEnv){\neg e_1, e}$}
			\bottomAlignProof
			\DisplayProof
			\quad
			\AxiomC{$\SequentCLJ(\ThAssumsEnv){e_1, e}$}
			\RightLabel{\LJNegL}
			\UnaryInfC{$\SequentCLJ(\ThAssumsEnv, \neg e_1){e}$}
			\bottomAlignProof
			\DisplayProof
			\\
			\AxiomC{$\SequentCLJ{e_1}$}
			\AxiomC{$\SequentCLJ{e_2}$}
			\RightLabel{\LJConjR}
			\BinaryInfC{$\SequentCLJ{e_1 \conj e_2}$}
			\bottomAlignProof
			\DisplayProof
		        \quad
			\AxiomC{$\SequentCLJ(\ThAssumsEnv, e_i){e}$}
			\AxiomC{$i \in \{1, 2\}$}
			\RightLabel{\LJConjL}
			\BinaryInfC{$\SequentCLJ(\ThAssumsEnv, e_1 \conj e_2){e}$}
			\bottomAlignProof
			\DisplayProof
			\quad
			\AxiomC{$\SequentCLJ(\ThAssumsEnv, e_1){e_2}$}
			\RightLabel{\LJImplR}
			\UnaryInfC{$\SequentCLJ{e_1 \impl e_2}$}
			\bottomAlignProof
			\DisplayProof
			\\
			\AxiomC{$\SequentCLJ(\ThAssumsEnv, e_2){e}$}
			\AxiomC{$\SequentCLJ(\ThAssumsEnv){e_1}$}
			\RightLabel{\LJImplL}
			\BinaryInfC{$\SequentCLJ(\ThAssumsEnv, e_1 \impl e_2){e}$}
			\bottomAlignProof
			\DisplayProof
			\quad
			\AxiomC{$\SequentCLJ{e}$}
			\AxiomC{$x \not\in FV(\ThAssumsEnv)$}
			\RightLabel{\LJAllR}
			\BinaryInfC{$\SequentCLJ{\all{x:\tau} e}$}
			\bottomAlignProof
			\DisplayProof
			\quad
			\AxiomC{$\SequentCLJ{e\subst{w/x}}$}
			\RightLabel{\LJExR}
			\UnaryInfC{$\SequentCLJ{\exist{x:\tau}e}$}
			\bottomAlignProof
			\DisplayProof
			\\
			\AxiomC{$\SequentCLJ(\Gamma_T, e){e_1}$}
			\AxiomC{$x \not\in FV(\ThAssumsEnv)$}
			\RightLabel{\LJExL}
			\BinaryInfC{$\SequentCLJ(\Gamma_T,\exist{x:\tau}e){e_1}$}
			\bottomAlignProof
			\DisplayProof
			\quad
			\AxiomC{$\SequentCLJ(\ThAssumsEnv, e_1 \subst{w/x}){e_2}$}
			\RightLabel{\LJAllL}
			\UnaryInfC{$\SequentCLJ(\ThAssumsEnv, \all{x:\tau}e_1){e_2}$}
			\bottomAlignProof
			\DisplayProof
			\end{empheq}
	\end{spreadlines}}
	\vspace*{-1em}
	
	\caption{\emph{\footnotesize{The rules for \LJ{}, standard structural rules are assumed (see Appendix~\ref{appendix:LJ-struct-rules}). In the rules we use $w$ to denote any closed expression that is well-typed with an empty context $\Delta$.
	}}}
	\label{fig:rules-CLJ}
	\vspace*{-1.5em}
\end{figure}


We define the set of \emph{FOL formulae} by induction on the formula shape, in a standard way (see Figure~\ref{fig:terms}). FOL formulae are a subset of well-typed expressions of LDL.
We use $\Gamma_T$ to denote a set of FOL formulae.
The rules in Figure~\ref{fig:rules-CLJ} follow the standard formulation of \LJ{}~\cite{sorensen2006lectures} (including notation $\Gamma_T, \psi$ for $\Gamma_T \cup \{\psi\}$). 
The equivalence closure of the reduction relation (convertibility) is denoted by $\conv$.
We will assume that for each comparison relation $\bowtie$ in the language, we have a corresponding oracle $\bowtie^*$ that decides, for any pair  $r_1, r_2$ of real numbers whether $r_1 \bowtie^* r_2$ holds. Formally, $ r_1 \bowtie^* r_2$ returns $\top$ if the relation holds, and $\bot$ otherwise. It is standard in intuitionistic logic to model negation $\neg e$ as $e \impl \bot$, but we use an equivalent formulation from~\cite{sorensen2006lectures} that introduces intuitionistic 
negation explicitly in the rules. 
 The only additional rule we need here is \textbf{(Arith)}. 
We assume the standard structural rules for \LJ{} which can be found in Appendix~\ref{appendix:LJ-struct-rules}.

Given this we can now define what it means for a differentiable logic $L$ to be \emph{sound} and \emph{complete} with respect to a set of formulae provable in \LJ.

\begin{definition}[Soundness]\label{df:sound}
A logic $L$ is \emph{sound} if for any well-typed formula $\validForm{e}$, and any contexts $\hNetCtx \in \ttype{\Xi}$, $\hRandCtx \in \ttype{q(e)}$, $\hVarCtx \in \ttype{\Delta}$ if $\tempty{e}_L^{\hNetCtx,\hRandCtx, \hVarCtx} = \tempty{\top}_L^{\hNetCtx,\hRandCtx, \hVarCtx}$ then $\SequentCLJ(){e}$.
\end{definition}

\begin{definition}[Completeness]\label{df:compl}
A logic $L$ is \emph{complete} if for any well-typed formula $\validForm{e}$, and any contexts $\hNetCtx \in \ttype{\Xi}$, $\hRandCtx \in \ttype{q(e)}$, $\hVarCtx \in \ttype{\Delta}$ if $\SequentCLJ(){e}$ then $\tempty{e}_L^{\hNetCtx,\hRandCtx, \hVarCtx} = \tempty{\top}_L^{\hNetCtx,\hRandCtx, \hVarCtx}$. 
\end{definition}

\subsubsection{Most Fuzzy DLs are sound but incomplete}
\label{subsec:soundness}

We start with showing that fuzzy DL comparison operators are valid. 

\begin{lemma}[Soundness of FDL comparisons]\label{lem:eval} If  $\validForm{e_1 \bowtie e_2}$ and any contexts $\hNetCtx \in \ttype{\Xi}$, $\hRandCtx \in \ttype{q(e)}$, $\hVarCtx \in \ttype{\Delta}$, then for all $L$ in fuzzy differentiable logics (FDL) the following holds:

If $\tempty{e_1 \bowtie e_2}^{\hNetCtx,\hRandCtx, \hVarCtx}_{L} = \tempty{\top}^{\hNetCtx,\hRandCtx, \hVarCtx}_{L}$ then $\tempty{e_1}^{\hNetCtx,\hRandCtx, \hVarCtx}_{L} \bowtie^* \tempty{e_2}^{\hNetCtx,\hRandCtx, \hVarCtx}_{L}  = \top$. 

If $\tempty{e_1 \bowtie e_2}^{\hNetCtx,\hRandCtx, \hVarCtx}_{L} = \tempty{\bot}^{\hNetCtx,\hRandCtx, \hVarCtx}_{L}$ then $\tempty{e_1}^{\hNetCtx,\hRandCtx, \hVarCtx}_{L} \bowtie^* \tempty{e_2}^{\hNetCtx,\hRandCtx, \hVarCtx}_{L}  = \bot$. 
\end{lemma} 

\begin{proof}
See Appendix~\ref{app:lemma}. This result uses Lemma~\ref{lem:expr-real}.
\end{proof}

\begin{lemma}
	\label{l:s1}
Given a well-typed formula $e$, for any any contexts $\hNetCtx \in \ttype{\Xi}$, $\hRandCtx \in \ttype{q(e)}$, $\hVarCtx \in \ttype{\Delta}$ the following hold:
\begin{enumerate}
	\item If $\tGodel{e}^{\hNetCtx,\hRandCtx, \hVarCtx} = \tGodel{\top}^{\hNetCtx,\hRandCtx, \hVarCtx}$   then $\SequentCLJ(){e}$. 
	\item If $\tGodel{e}^{\hNetCtx,\hRandCtx, \hVarCtx} = \tGodel{\bot}^{\hNetCtx,\hRandCtx, \hVarCtx}$ 	then $\SequentCLJ(e){}$.
\end{enumerate}
\end{lemma}

\begin{proof}
The auxiliary lemma is proven by case analysis on $e$, mutual induction, and relies on Lemma~\ref{lem:eval}, see Appendix \ref{app:proof-godel}.
\end{proof}

\begin{theorem}[Soundness of G\"{o}del DL]\label{th:s}
	G\"{o}del DL is sound.
\end{theorem}

\begin{proof}
Soundness is obtained as a corollary of Lemma~\ref{l:s1}.
\end{proof}
%
%


\noindent Other fuzzy DLs, with the exception of Łukasiewicz and Yager are also sound with respect to \LJ, with proofs following the same scheme. For the Łukasiewicz DL, we cannot prove the equivalent of Lemma~\ref{l:s1}. 
Intuitively, this is because Łukasiewicz implication is not strong enough: e.g. $ \tempty{e_1 \impl e_2} $ evaluates to $1$ as long as $\tempty{e_1} \leq \tempty{e_2}$. Because $\tempty{e_1}$ and $\tempty{e_2}$ can take any values in the interval $[0,1]$ (as long as $\tempty{e_1} \leq \tempty{e_2}$), it is  
impossible to apply an inductive argument on $\tempty{e_1}$ and $\tempty{e_2}$. 
The Yager DL is a generalisation of the Łukasiewicz DL, in particular, when $p = 1$ in the Yager DL translation, it coincides with the Łukasiewicz DL. Thus, the same problem stands for Yager.  

Proper study of the question of (alternative approaches to) soundess of the Łukasiewicz and Yager DLs  is left for future work.  We note that the recent work by \citet{BMPP23} has shed light on logical properties of Łukasiewicz logic, and we intend to build upon those results.

Theorem~\ref{th:s} does not hold for the opposite direction of implication (and thus completeness fails).
  To see this, consider the following derivation for 
$\SequentCLJ{\neg(5 \leq 3)}$ in \LJ, which we would have liked to be equivalent to having $\tGodel{\neg(5 \leq 3)} \in \tGodel{\top}$. 

{\footnotesize{
\begin{prooftree}
  \AxiomC{}
      \RightLabel{\rulelabel{$\bot$}}
      \UnaryInfC{$\Gamma_T, \bot \vdash$}
      \RightLabel{\rulelabel{$\mathbf{Arith-L}$}}
      \UnaryInfC{$\Gamma_T, 5 \leq 3 \vdash$}
      \RightLabel{\rulelabel{$\mathbf{\neg-R}$}}
      \UnaryInfC{$\Gamma_T \vdash \neg(5 \leq 3)$}
    \end{prooftree}}}

\noindent However, taking $\tGodel{5 \leq 3} = 1 - 0.25 = 0.75$, we get 
$\tGodel{\neg(5 \leq 3)} = 1 - 0.75 = 0.25 \notin \tGodel{\top}$. This is the problem that will be common for all Fuzzy DLs.


\subsubsection{DL2 is sound and incomplete}\label{subsec:soundness2}

LDL helps to establish a generic approach to proving soundness results for a variety of DLs.
Firstly, a variant of Lemma~\ref{lem:eval} holds for DL2. 
 Next we obtain its soundness for DL2, where DL2 has just the connectives $\land, \lor$ and quantifiers, as $\impl$ and $\neg$ are not defined.

\begin{theorem}[Soundness of DL2]\label{th:d1}
DL2 is sound. 
\end{theorem}
\begin{proof}
See Appendix~\ref{app:dl2-adequacy}.
\end{proof}

The proof 
follows the structure of similar proofs for FDLs. This is a welcome simplification, and it will be useful in future computer formalisations of these logics. Seeing the result is given for an incomplete set of connectives, we omit discussion of completeness for DL2.  

\subsubsection{STL is neither sound nor complete}\label{subsec:STL-S}
Incompleteness of STL is discussed in~\S\ref{subsec:soundness}.
In addition, note that its connectives and quantifiers lack properties that are provable in \LJ: e.g.~associativity of conjunction, or commutativity of a universal quantifier with respect to conjunction. 


The definition for soundness relies on the interpretation of $\tSTL{\top}^{\hNetCtx,\hRandCtx, \hVarCtx}$.  \citet{varnai} do 
	not formally define $\tSTL{\top}^{\hNetCtx,\hRandCtx, \hVarCtx}$,
	and for that reason alone soundness is unattainable for the original STL. 
	 We show an attempt to propose an interpretation $\tSTL{\top}^{\hNetCtx,\hRandCtx, \hVarCtx} = \infty$ and show how we use LDL's generic approach to analyse the result. 
	Informally, their intuition is that any positive number in $[0, \infty)$
	 belongs to $\tSTL{\top}^{\hNetCtx,\hRandCtx, \hVarCtx}$. This motivates, for example,
	$\tSTL{ 3 \leq 5}^{\hNetCtx,\hRandCtx, \hVarCtx} = 2$ or  $\tSTL{ 3 == 3}^{\hNetCtx,\hRandCtx, \hVarCtx} = 0$. But none of these evaluates to $\infty$. Thus, although the soundness proof goes through, the set of expressions that a variant of Lemma~\ref{lem:eval} for STL covers will be empty: such soundness would say nothing about the FOL fragment of LDL.  A solution would be to re-define interpretation for all predicates in a binary fashion: e.g. \emph{for $ a_1 == a_2$, return $\infty$ if $a_1 = a_2$ and return $-\infty$ otherwise}.
	But this would sacrifice continuity and smoothness of the resulting loss functions, -- key properties for STL design.

\subsection{Logical and Geometric Properties of DLs}
\label{sec:properties}

 Both \citet{varnai} and \citet{VANKRIEKEN2022103602} suggest a selection of desirable properties for their DLs. Some are logical properties (e.g. associativity or commutativity) and some have a geometric origin, such as (weak) smoothness or shadow-lifting.    
   As some of the DLs we consider do not have associative connectives, we will use $\conjM$ as a notation for conjunction of exactly $ M $ conjuncts. Similar definition can be made for disjunction, which we omit. 
 In all of the following definitions we will omit the contexts for clarity as they do not change. We will denote a well-typed formula by $e$.
 Starting with geometric properties, 
smoothness is generally desirable from optimisation perspective as it aids gradient based methods used in neural network training~\cite{LeeRY23}:

  \begin{definition}[Weak smoothness]
	The $\tempty{e}_L$ is \textit{weakly smooth} if it is continuous everywhere and its gradient is continuous at all points in the interval where there is a unique minimum.
\end{definition}

\begin{definition}[Scale invariance for $\conjM$]
	The $\tempty{e}_L$ is \textit{scale-invariant} if, for any real $ \alpha \leq 0$ 
	\begin{equation*}
	\alpha \tempty{\conjM (A_1, ..., A_M)}_L = \tbig{\conjM}_L (\alpha \tempty{A_1}_L, ..., \alpha \tempty{A_M}_L)
	\end{equation*}
\end{definition}

\emph{Shadow-lifting} is a property original to~\citet{varnai} and motivates the STL conjunction. 
It characterises gradual improvement when training the neural network: 
if one conjunct increases, the value of entire conjunction should increase as well. 

\begin{definition}[Shadow-lifting property for $\conjM$]
	The $\tempty{e}_L$ satisfies the \textit{shadow-lifting property} if, $\forall i. \tempty{A_i}_L \neq 0$:
	
	\begin{equation*}
	\left. \dfrac{\partial \tempty{\conjM(A_1, ..., A_M)}_L}{\partial \tempty{A_i}_L}\right\rvert_{A_1, ..., A_M} >0
	\end{equation*}
	where $ \partial $ denotes partial differentiation.
\end{definition}

Coming from the logic perspective, we have the following desirable properties: 

\begin{definition}[Commutativity, idempotence and associativity of  $\conjM$]
	The $\tempty{e}_L$ is \emph{commutative} if for any permutation $\pi$ of the integers $i \in {1, ..., M}  $
	\begin{equation*}
	\tempty{\conjM (A_1, ..., A_M)}_L = \tempty{\conjM (A_{k_{\pi(1)}}, ..., A_{k_{\pi(M)}})}_L
	\end{equation*}
	it is \emph{idempotent} and \emph{associative} if
	\begin{align*} 
	\tempty{\conjM(A, ..., A)}_L &= \tempty{A}_L\\
	\tempty{\conjL{2}(\conjL{2}(A_1, A_2), A_3)}_L &= \tempty{\conjL{2}(A_1, \conjL{2}(A_2, A_3))}_L
	\end{align*} 
	
\end{definition}



The below property has not appeared in the literature before, but as previous sections have shown, it proves to be important when generalising DLs to FOL: 

\begin{definition}[Quantifier commutativity]
	The $\tempty{e}_L$ satisfies the \emph{quantifier commutativity} if 
	\begin{equation*}
		\tempty{\forall x. \conjM(A_1, ..., A_M)}_L = \tempty{\conjM(\forall x. A_1, \ldots, \forall x. A_M)}_L
	\end{equation*}
	Similarly for $\exists$ and $\lor$. 
\end{definition}


\begin{table}[t]

	\caption{Three groups of properties of DLs: geometric, logical and semantical. Properties which have been established in other works have relevant citations. Entries with * denote DLs for which the smoothness property holds in propositional case.}
	\label{tab:properties}
	\begin{adjustbox}{center}
		\begin{tabular}{|p{0.28\textwidth}|c|c|c|c|c|c|}
			\hline 
			\textbf{Properties:}& DL2 &  G\"{o}del & Łukasiewicz& Yager & Product & STL\\ 
			\hline \hline
			Weak Smoothness & yes* & no &no &no &yes* & yes\\
			\hline
			Shadow-lifting  & yes & no &no&no&yes& yes \cite{varnai}\\
			\hline 
				Scale invariance & yes & yes &no&no&no& yes \cite{varnai}\\
			\hline \hline
			Idempotence & no & yes \cite{VANKRIEKEN2022103602} &no \cite{cintula2011handbook}&no \cite{klement2004triangular}&no \cite{cintula2011handbook}& yes \cite{varnai}\\
			\hline 
			Commutativity & yes & yes \cite{VANKRIEKEN2022103602} &yes \cite{VANKRIEKEN2022103602}&yes \cite{VANKRIEKEN2022103602}&yes \cite{VANKRIEKEN2022103602}& yes \cite{varnai}\\
			\hline 
					Associativity &yes &yes \cite{VANKRIEKEN2022103602}&yes \cite{VANKRIEKEN2022103602}&yes \cite{VANKRIEKEN2022103602}&yes \cite{VANKRIEKEN2022103602}&no \cite{varnai}\\
			\hline
			Quantifier commutativity & no & yes & no & no& no & no \\
			\hline 
			Soundness & yes & yes & no & no & yes & no \\
			\hline \hline
			\end{tabular}
		\\
	\end{adjustbox}		
	\setlength{\belowcaptionskip}{-15pt}
	\vspace*{-1em}
\end{table}


Table \ref{tab:properties} summarises how different DLs compare and contains results already established in literature (as cited) as well as provides new results. While we do not provide proofs, the reader can fine an in-depth explanation of the reasons behind the mentioned properties in Appendix~\ref{app:properties}.
 The table inspires the following observations. 



\begin{inparaenum}[i)]
\item Starting with geometric properties, smoothness is traditionally important for differentiation~\cite{LeeRY23}, and connectives for DL2 and Product DL satisfy the property. However, in the FOL extension only STL remains weakly smooth, see also discussion of \S~\ref{subsec:STL-S}.

\item For the remaining geometric properties,  DL2 and STL 
satisfy shadow-lifting and scale invariance, but Fuzzy logics disagree on these. 
Thus, only STL satisfies all geometric properties. 

\item For logical properties, \citet{varnai} suggested that a DL cannot both satisfy shadow-lifting, idempotence and associativity. However, their proof of this fact fails within LDL. We leave formal investigation of this conjecture for future work.

\item Since the semantics of quantifiers are defined as global minima and maxima, only connectives of G\"{o}del DL commute with respect to them.
This is a highly desirable property, that has not been previously discussed in the literature. It impacts completeness of DLs (cf. \S~\ref{subsec:soundness} and \ref{subsec:STL-S}).

\item Only G\"{o}del DL satisfies all logical properties, but it fails smoothness and shadow-lifting properties. On the opposite side of the spectrum, STL  behaves well geometrically, but fails associativity and quantifier commutativity, and has unresolved issues with soundness (\S~\ref{subsec:STL-S}). It warrants future research whether a DL with optimal combination of logical and geometric  properties may be formulated. 
\end{inparaenum}

\subsection{Empirical Evaluation}
\label{sec:eval}

LDL offers a unified framework to evaluate different DLs.
 We have implemented the syntax and semantics as an extension to the Vehicle tool ~\cite{Vehicle} which  supports the general syntax of LDL. The implementation is highly modular due to the semantics of each DL sharing the same semantics for the core of the syntax as explained in Section~\ref{sec:semantics}. We explain the workflow for our running example -- 
the robustness property. 

Taking the specification of Example~\ref{ex:wte}, it is compiled by Vehicle to a loss function in Python where it can be used to train a chosen neural network. The loss function resulting from the LDL semantics is typically used in combination with a standard loss function such as cross-entropy.
If we denote cross-entropy loss together with its inputs as $\mathcal{L}_{CE}$ and the LDL loss together with its inputs and parameters as $\mathcal{L}_{DL}$ then the final loss will be of the form:
$\mathcal{L} = \alpha \cdot \mathcal{L}_{CE} + \beta \cdot \mathcal{L}_{DL}$
where the weights $\alpha, \beta \in \Real$ are parameters.
For the preliminary\ tests we used a simple 2-layer network to classify MNIST images~\cite{deng2012mnist}. As seen in Figure~\ref{fig:experiments} we have tested the effect of varying the $\alpha, \beta \in \Real$ parameters and therefore the weight of the custom loss in training on accuracy and constraint satisfaction; the results are broadly consistent with trends reported in the literature previously. 

We note that formulating heuristics for finding minima/maxima of loss functions, which are necessary for the implementation of quantifier interpretation, can be very complex, and are beyond the scope of this paper.
Figure~\ref{fig:experiments} shows results implemented using random sampling instead of sampling around global maxima or minima.
We leave implementation of heuristics of finding global minima/maxima, as well as thorough experimental study, for future work.

\begin{figure}
	\centering
	\begin{subfigure}{.4\textwidth}
		\includegraphics[width=1\linewidth]{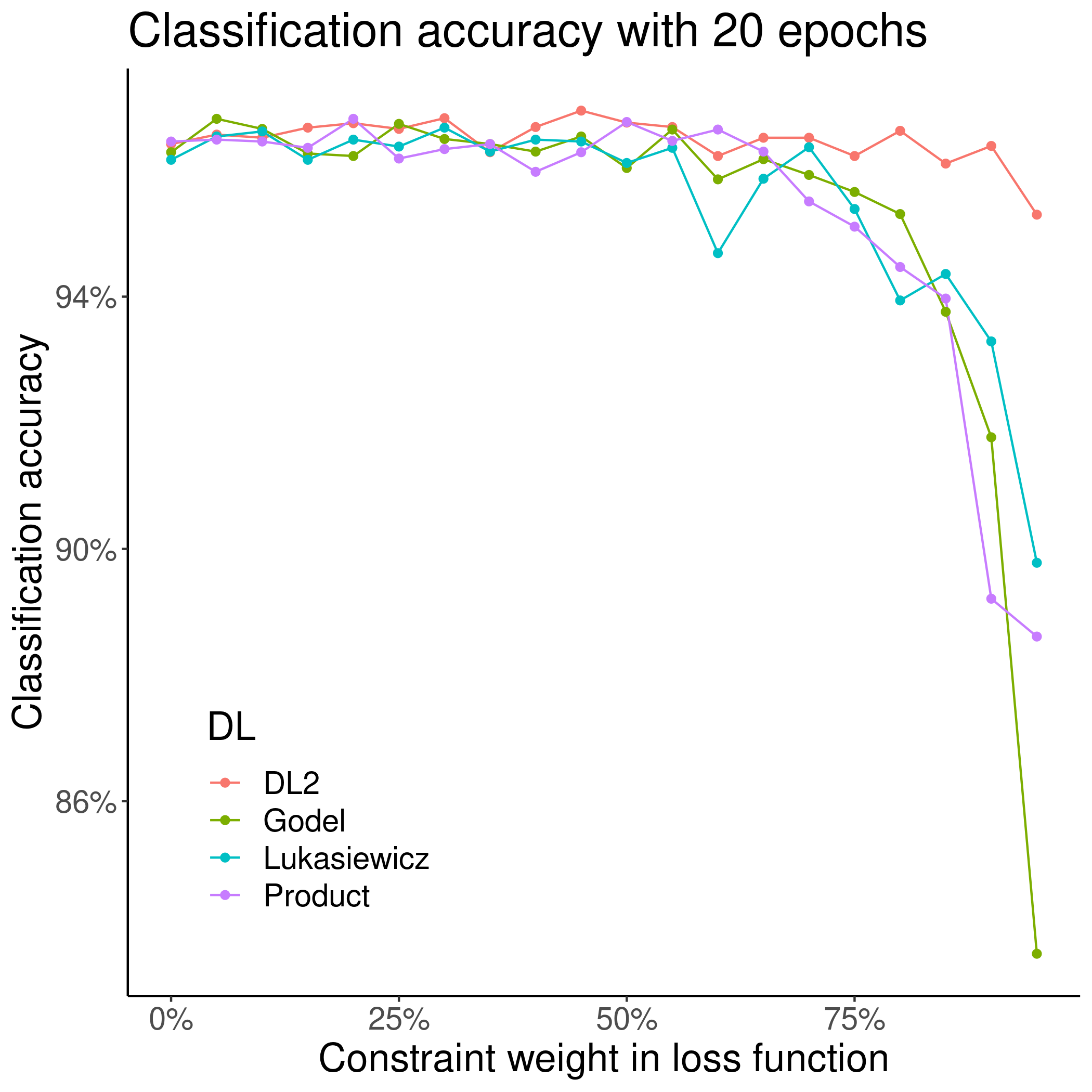}
	\end{subfigure}
		\begin{subfigure}{.4\textwidth}
		\includegraphics[width=1\linewidth]{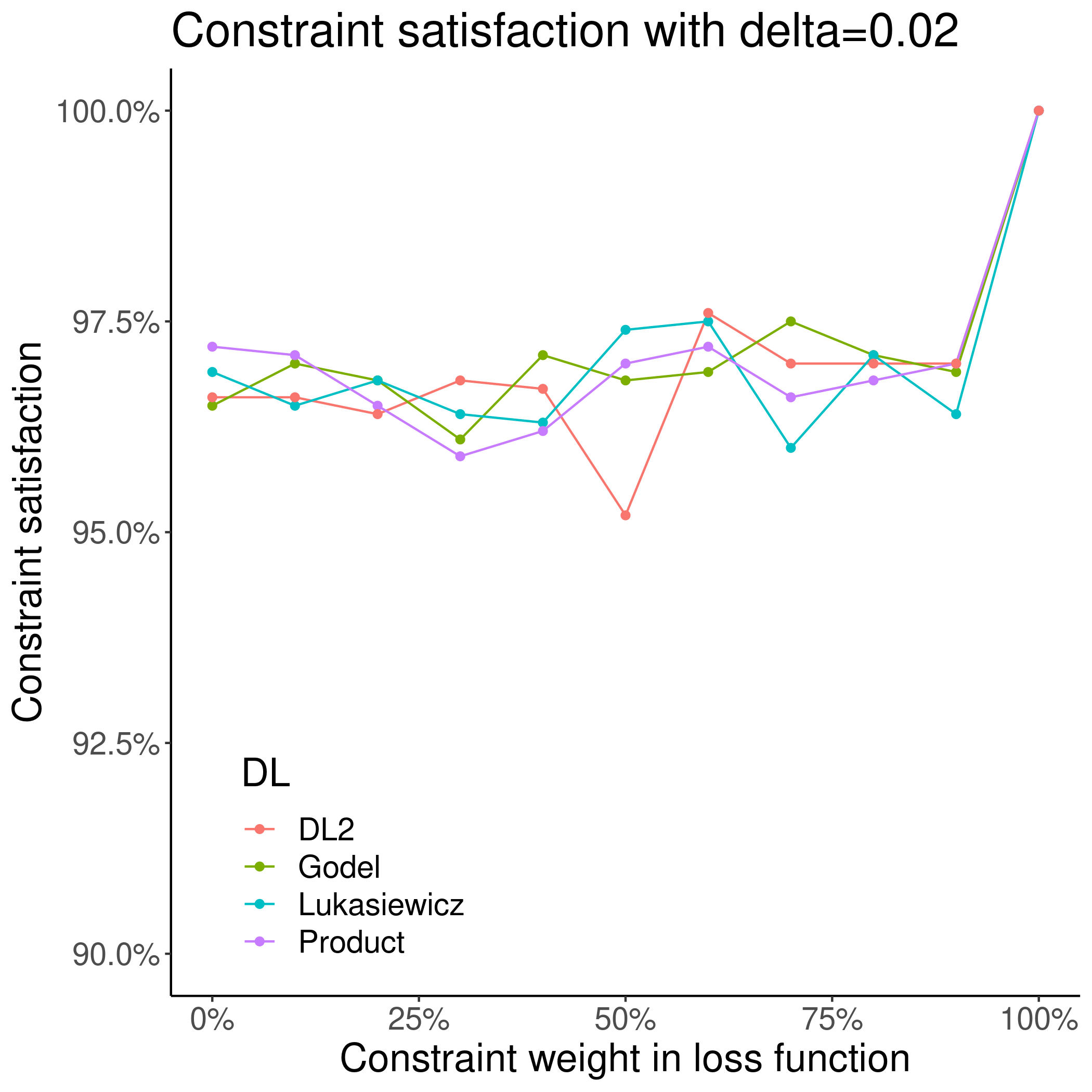}
	\end{subfigure}
\caption{Performance of the network trained with different custom losses with varying parameters $\alpha, \beta \in \Real$ measured by relative weight of $\beta$ in relation to $\alpha$.}
\label{fig:experiments}
	\vspace*{-1.5em}
\end{figure}

\section{Conclusions, Related and Future Work}

\paragraph{Conclusions.}
We have presented a general language (LTL) for expressing properties of neural networks. Our contributions are: 1) LDL  generalises other known DLs and provides a language that is rich enough to express complex verification properties; 2) with LDL we achieved the level of rigour that allows to formally separate the formal language from its semantics, and thus opens a way for systematic analysis of properties of different DLs; 3) by defining different DLs in LDL, we proved properties concerning their soundness and resulting loss functions, and opened the way for uniform empirical evaluation of DLs. We now discuss related work.

\paragraph{Learning Logical Properties.} 
There are many methods of passing external knowledge in the form of logical constraints to neural networks. The survey by~\citet{giunchiglia2022deep} discusses multitude of approaches including methods based on loss functions~\cite{pmlr-v80-xu18h,fischer2019dl,VANKRIEKEN2022103602}, to which LDL belongs, but also others such as tailoring neural network architectures~\cite{garcez2019neural}, or  guaranteeing constraint satisfaction at the level of neural network outputs~\cite{hoernle2022multiplexnet,dragone2021neuro}. As \citet{giunchiglia2022deep} establish, the majority of approaches are tailored to specific problems, and only \citet{fischer2019dl} and \citet{VANKRIEKEN2022103602} go as far as to include quantifiers. LDL generalises both.

\textbf{Analysis of properties of loss functions.} Property analysis, especially of smoothness~\cite{LeeRY23} or bilateral properties~\cite{nie2018investigation}, is a prominent field~\cite{wang2022comprehensive}. One of LDL's achievements is to expose trade-offs between satisfying desired geometric and logic properties of a loss functions.

\textbf{Neural Network Verification.} While this work does not attempt to \emph{verify} neural networks, we draw our motivation from this area of research. It has been observed in verification literature that neural networks often fail to satisfy logical constraints~\cite{wang2018efficient}. One of proposed solutions is training the NN to satisfy a constraint prior to verifying them~\cite{hu2016harnessing,pmlr-v80-xu18h}. This belongs to an approach referred to as \emph{continuous verification}~\cite{KKK20,CKDKKAE22} which focuses on the cycle between training and verification. LDL fits into this trend. Indeed, the tool Vehicle that implements LDL is also built to work with NN verifiers~\cite{daggitt2023compiling}.  

\textbf{Logics for Uncertainty and Probabilistic Logics.} LDLs have a strong connection to fuzzy logic~\cite{VANKRIEKEN2022103602}. Via the use of probability distributions and expectations, we draw our connection to Probabilistic Prolog and similar languages~\cite{de2015probabilistic,10.5555/1625275.1625673}. We differ as none of those approaches can be used to formulate loss functions, which is the main goal of LDL.    

\textbf{Adversarial training.} Starting with the seminal paper by~\citet{szegedy2013intriguing}, thousands of papers in machine learning literature have been devoted to adversarial attacks and robust training of neural networks. Majority of those papers does not use a formal logical language for attack or loss function generation. LDL opens new avenues for this community, as allows one to formulate other properties of interest apart from robustness, see e.g.~\cite{CKDKKAE22,fomlas22}.

   
\paragraph{Future work.}
We intend to use LDL to further study the questions of a suitable semantics for DL, both proof-theoretic and denotational, possibly taking inspiration from~\cite{BMPP23}. For proof-theoretic semantics, we need to find new DLs with tighter correspondence to \LJ; and moreover new calculi can be developed on the basis of \LJ ~to suit the purpose. Since loss functions are used in computation, we conjecture that constructive logics (or semi-constructive logics as in~\cite{BMPP23}) will be useful in this domain. For denotational semantics, we intend to explore Kripke frame semantics.  
Future work could also include finding the best combination of logical and geometric properties in a DL and formulating new DLs that satisfy them. Thorough evaluation of performance of all the DLs is also left to future work: it includes formulating heuristics for finding local/global maxima for quantifier interpretation, and evaluation of LDL effectiveness within the continuous verification cycle. Finally,
 finding novel ways of defining quantifiers that commute with DL connectives is an interesting challenge.

\section{Acknowledgements}
This work was supported by the EPSRC grant EP/T026952/1, \emph{AISEC: AI Secure and Explainable by Construction} and the EPSRC DTP Scholarship for N.~\'{S}lusarz. We thank anonymous referees, James McKinna, Wen Kokke, Bob Atkey and Emile Van Krieken for valuable comments on the early versions of this paper, and Marco Casadio for contributions to the Vehicle implementation. 

\bibliographystyle{plainnat}
\bibliography{bibliography,Natalia}
\newpage
\appendix

\section{Supplementary Definitions for Different DLs}
\subsection{Additional DLs based on Fuzzy logic}
\label{appendix:additional-dls}

LDL can be used to express other DLs based on fuzzy logic aside from ones defined in Section~\ref{sec:semantics}. Figure~\ref{fig:expr1} gives a few more examples originally from~\citet{VANKRIEKEN2022103602}.
The semantics of comparison operators are omitted as they are identical to those of the \godel ~logic.

\begin{figure}[t]
	\footnotesize
	\renewcommand{\arraystretch}{1.3}
	\centering
		\begin{tabular}{|c|c|c|c|}
			\hline 
			Syntax & 
			Łukasiewicz & 
			Yager &
			product \\ 
			\hline 
			$\tempty{\BoolType}_L$ & 
			$[0, 1]$ &
			$ [0, 1] $ & 
			$ [0, 1] $
			\\ \hline
			$\tempty{\top}_L$ & 
			1 & 
			1 &
			1
			\\ \hline
			$\tempty{\bot}_L$ & 
			0 &
			0 &
			0
			\\ \hline
			$\tempty{\neg}_L$ & 
			$\lam{\val}1 - \val$ &
			$\lam{\val}1 - \val$ &
			$\lam{\val}1 - \val$
			\\ \hline 
			$\tempty{\wedge}_L$ & 
			$ \lam{\val_1,\val_2} \max(\val_1+\val_2-1,0)) $ &
			$ \lam{\val_1,\val_2}\max (1 - ((1-\val_1)^p + (1-\val_2)^p)^{1/p},0)$ &
			 $\lam{\val_1,\val_2} \val_1 \cdot \val_2$
			\\ \hline 
			$\tempty{\vee}_L$& 
			$ \lam{\val_1,\val_2} \min (\val_1+\val_2,1)  $ &
			$ \lam{\val_1,\val_2}\min ((\val_1^p + \val_2^p)^{1/p},1)$ & 
			$\lam{\val_1,\val_2} \val_1 + \val_2 - \val_1 \times \val_2$
			\\ \hline 
			$\tempty{\impl}_L$ & 
			$ \lam{\val_1,\val_2} \min (1 -\val_1+\val_2,1)  $ & 
			 - &
			$ \lam{\val_1,\val_2} 1 - \val_1 + \val_1 \times \val_2 $
			\\ \hline
			$ \tempty{==}_L $ &
			\coloured{$\lambda \val_1, \val_2. 1 -  \tanh |\val_1-\val_2|$} &
			\coloured{$\lambda \val_1, \val_2. 1 -  \tanh |\val_1-\val_2|$} &
			\coloured{$\lambda \val_1, \val_2. 1 -  \tanh |\val_1-\val_2|$}
			\\ \hline
			$ \tempty{\leq}_L $ &
			\coloured{$\lam{\val_1, \val_2}{1-\max(\tanh |\val_1-\val_2|, 0)}$} &
			\coloured{$\lam{\val_1, \val_2}{1-\max(\tanh |\val_1-\val_2|, 0)}$} &
			\coloured{$\lam{\val_1, \val_2}{1-\max(\tanh |\val_1-\val_2|, 0)}$}
			\\ \hline
		\end{tabular}
	\caption{Semantics of LDL expressions dependent on the choice of DL. \coloured{Colour} denotes parts of semantics added for LDL and not defined originally. In the implication row ``-'' denotes that the semantics of implication was not provided separately and is instead defined with negation and disjunction in a standard manner. In the Yager DL, $p\geq 1$ is a parameter.}
	\label{fig:expr1}
\end{figure}

\subsection{Remaining comparisons}
\label{app:comparisons}

We have defined the semantics of both $\leq$ and $==$ in Section~\ref{subsec:expressions}. In many logics the semantics of other comparisons could be expressed in terms of the semantics of $\leq$, $==$ and logical connectives. However, as not all DLs presented have negation, this approach is not feasible in general and therefore we define all the remaining comparisons in Figure~\ref{fig:more-comparisons}.

\begin{figure}
	\footnotesize
	\centering
	\renewcommand{\arraystretch}{1.3}
	\begin{tabular}{|c|c|c|c|}
		\hline 
		Syntax & DL2 & Fuzzy Logics & STL \\ 
		\hline 

		$\tempty{\neq}_{L}$ &
		$ \lam{\val_1,\val_2}{-\xi[\val_1=\val_2]}  $&
		\coloured{$ \lam{\val_1,\val_2}{1-[\val_1=\val_2]}  $}&
		\coloured{$ \lam{\val_1,\val_2}{-\xi[\val_1=\val_2]}  $}\\
		\hline
		$ \tempty{<}_{L} $ &
		$\lam{\val_1, \val_2} \tDLtwo{(\val_1\leq\val_2)\land (\val_1 \neq \val_2)}$ &
		\coloured{$\lam{\val_1, \val_2} \tempty{(\val_1\leq\val_2)\land (\val_1 \neq \val_2)}_{FDL}$}&
		\coloured{$\lam{\val_1, \val_2} \tSTL{(\val_1\leq\val_2)\land (\val_1 \neq \val_2)}$}\\
		\hline
		$ \tempty{\geq}_{L} $ &
		$\lam{\val_1, \val_2}{\val_2 \leq \val_1}$ &
		\coloured{$\lam{\val_1, \val_2}\val_2 \leq \val_1$} &
		\coloured{$\lam{\val_1, \val_2}\val_2 \leq \val_1$}
		\\
		\hline
		$ \tempty{>}_{L} $ &
		$\lam{\val_1, \val_2} \tDLtwo{(\val_1\geq\val_2)\land (\val_1 \neq \val_2)}$&
		\coloured{$\lam{\val_1, \val_2} \tempty{(\val_1\geq\val_2)\land (\val_1 \neq \val_2)}_{FDL}$}&
		\coloured{$\lam{\val_1, \val_2} \tSTL{(\val_1\geq\val_2)\land (\val_1 \neq \val_2)}$}\\
		\hline
	\end{tabular}
\caption{Semantics of comparisons between terms in LDL that are dependant on the choice of DL. Interpretations in \coloured{colour} have been defined for the purposes of LDL and were not present in the original semantics.  $\xi > 0$ is a constant and $[\cdot]$ the indicator function.}
\label{fig:more-comparisons}
\end{figure}

\subsection{Negation in DL2}
\label{ap:neg}

DL2 does not have an interpretation for negation - instead negation of a term is pushed inwards syntactically to the level of comparisons between terms. While this cannot be expressed as a single lambda function, it is possible to express it as a partial function that applies the same operation on LDL syntax, as defined in Figure~\ref{eq:lowerNot}. It is however not possible to define a total function as it is not possible to push the negation through parts of syntax such as a lambda.

\begin{figure}
	\begin{align*}
	\tDLtwo{\neg(a_1 \wedge a_2)}	& =  \tDLtwo{\neg a_1 \vee \neg a_2}\\
	\tDLtwo{\neg(a_1 \vee a_2)}	& =  \tDLtwo{\neg a_1 \wedge \neg a_2}\\
	\tDLtwo{\neg(a_1 \leq a_2)}	& = \tDLtwo{a_1 > a_2 }\\
	\tDLtwo{\neg(a_1 < a_2)}	& = \tDLtwo{a_1 \geq a_2 }\\
	\tDLtwo{\neg(a_1 \geq a_2)}	& = \tDLtwo{a_1 < a_2 }\\
	\tDLtwo{\neg(a_1 > a_2)}	& = \tDLtwo{a_1 \leq a_2 }\\
	\tDLtwo{\neg(a_1 == a_2)}	& = \tDLtwo{a_1 \neq a_2 }\\
	\tDLtwo{\neg(a_1 \neq a_2)}	& = \tDLtwo{a_1 == a_2 }
	\end{align*}
\caption{Definition of pushing the negation inwards present in the DL2 based DL.}
\label{eq:lowerNot}
\end{figure}

\section{Supplementary Definitions for \LJ.}

\subsection{Structural Rules for LJ}
\label{appendix:LJ-struct-rules}

\LJ{} has four structural rules: weakening on the left {\footnotesize{\LJWeakLT{}}}, exchange on the left {\footnotesize{\LJExchLT{}}}, contraction on the left {\footnotesize{\LJCtrLT{}}}, and weakening on the right {\footnotesize{\LJWeakRT{}}}. 

{\footnotesize
	\begin{gather*}
	\AxiomC{$\SequentCLJ(\ThAssumsEnv){e}$}
	\UnaryInfC{$\SequentCLJ(\ThAssumsEnv, e_1){e}$}
	\bottomAlignProof
	\DisplayProof
	\qquad
	\AxiomC{$\SequentCLJ{e}$}
	\UnaryInfC{$\SequentCLJ{e, e_1}$}
	\bottomAlignProof
	\DisplayProof
	\qquad
	\AxiomC{$\SequentCLJ(\ThAssumsEnv, e_1, e_2, \Gamma_T'){e}$}
	\UnaryInfC{$\SequentCLJ(\ThAssumsEnv, e_2, e_1, \Gamma_T'){e}$}
	\bottomAlignProof
	\DisplayProof
	\qquad
	\AxiomC{$\SequentCLJ(\ThAssumsEnv, e_1, e_1){e}$}
	\UnaryInfC{$\SequentCLJ(\ThAssumsEnv, e_1){e}$}
	\bottomAlignProof
	\DisplayProof
	\end{gather*}}

\section{Proof of Type Soundness of LDL}
\label{ap:type-soundness}

We prove by induction on the typing judgment the Theorem~\ref{th:type-sound} which reads as follows.

\begin{theorem*}[Type Soundness of LDL]
For all differentiable logics $L$ and well typed expressions $\hcTypeRel{e}{\tau}$, then for all $\hNetCtx \in \ttype{\Xi}$, $\hRandCtx \in \ttype{q(e)}$ and $\hVarCtx \in \ttype{\Delta}$
 we have $\tempty{e}^{\hNetCtx,\hRandCtx,\hVarCtx}_L \in \ttype{\tau}$.
\end{theorem*}

\begin{proof}
\textbf{Base Case 1.} 
Suppose we have $\hcTypeRel{e}{\VecType{m} \to \VecType{n}}$ ($e$ is a network variable). Then we have $e: \VecType{m} \to \VecType{n} \in \Xi$. But then we have $N[e] \in (\Real^n)^{\Real^m}$ by assumption.

\textbf{Base Case 2.}
Suppose we have $\hcTypeRel{e}{\tau}$ ($e$ is a bound variable). Then we have $e : \tau \in \Delta$. But then we have $\Gamma[e] \in \ttype{\tau}$ by the definition of $\Gamma$.

\textbf{Base Case 3.}
Suppose we have $\hcTypeRel{\elReal}{\RealType}$. By definition in Figure~\ref{fig:types} we also have $r \in \Real$.

\textbf{Base Case 4.}
Suppose we have $\hcTypeRel{i}{\FinType}$. By definition in Figure~\ref{fig:types} we also have $i \in \text{Index } n$.

\textbf{Base Case 5.}
Suppose we have $\hcTypeRel{b}{\BoolType}$. By definition in Figure~\ref{fig:types} we have $b \in \ttype{\BoolType}$. The interpretation of $\ttype{\BoolType}$ depends on the DL in question but is always a subset of $\Real$ and does not impact the proof.

\textbf{Base Case 6.}
Suppose we have $\hcTypeRel{e}{+ :\RealType \to \RealType \to \RealType}$. This follows by definition in Figure~\ref{fig:types}.

\textbf{Base Case 7.}
Suppose we have $\hcTypeRel{e}{\times :\RealType \to \RealType \to \RealType}$. This follows by definition in Figure~\ref{fig:types}.

\textbf{Base Case 8.} Suppose we have $\hcTypeRel{e}{\wedge,\vee,\implies :\BoolType \to \BoolType \to \BoolType}$. This follows by definition in Figure~\ref{fig:types}.

\textbf{Base Case 9.} $\hcTypeRel{e}{\not :\BoolType \to \BoolType}$ This follows by definition in Figure~\ref{fig:types}.

\textbf{Base Case 9.} $\hcTypeRel{e}{\bowtie :\RealType \to \RealType \to \BoolType}$ This follows by definition in Table~\ref{fig:types}.

\textbf{Inductive Case 1.}
Suppose $e$ is an application and we have $\hcTypeRel{\App{e_1}{e_2}}{\tau_2}$.
 By definition in Figure~\ref{fig:types} we have $ \hcTypeRel{e_1}{\FunType{\tau_1}{\tau_2}} $ and $ \hcTypeRel{e_2}{\tau_1} $.
 Then by the induction hypothesis we have $\tempty{e_1}^{\hNetCtx,\hRandCtx,\hVarCtx}_L \in \ttype{\tau_2}^{\ttype{\tau_1}}$ and $\tempty{e_2}^{\hNetCtx,\hRandCtx,\hVarCtx}_L \in \ttype{\tau_1}$.
 By application we directly have that $ \tempty{e_1}^{\hNetCtx,\hRandCtx,\hVarCtx}_L \tempty{e_2}^{\hNetCtx,\hRandCtx,\hVarCtx}_L \in \ttype{\tau_2}$.
 
 \textbf{Inductive Case 2.}
 Suppose we have $\hcTypeRel{\Lam{\id}{\tau_1}{e}}{\FunType{\tau_1}{\tau_2}}$ ($e$ is a lambda).
 Then by  definition in Figure~\ref{fig:types} we have  $ \hTypeRel{\hcNetCtx,\consNew{\id \rightarrow \tau_1}{\hcVarCtx}}{e}{\tau_2} $. 
 We want to show that $\lambda y . \tempty{e}^{\hNetCtx,\hRandCtx,\hVarCtx[x\rightarrow y]}_L \in \ttype{\tau_2}^{\ttype{\tau_1}}$.
 	 Assuming $y \in \ttype{\tau_1}$ we therefore need to show that $\tempty{e}^{\hNetCtx,\hRandCtx,\hVarCtx[x\rightarrow y]}_L \in \ttype{\tau_2}$.
 By the induction hypothesis we can now show that $\forall z:\tau \in \hcVarCtx[(x\rightarrow\tau_1)]\ .\ (\hVarCtx[x\rightarrow \tempty{e_1}^{\hNetCtx,\hRandCtx,\hVarCtx}_L])[z] \in \ttype{\tau}$.

\textbf{Inductive Case 3.}
Suppose we have $\hcTypeRel{\text{let} (x\ :\ \tau_1) = e_1\ \text{in}\ e_2}{\tau_2}$ ($e_1$ is a let).
 Then by  definition in Figure~\ref{fig:types} and the induction hypothesis we have $\tempty{e_1}^{\hNetCtx,\hRandCtx,\hVarCtx}_L \in \ttype{\tau_1}$
 and that $\hTypeRel{\hcNetCtx,\consNew{\id,\tau_1}{\hcVarCtx}}{e_2}{\tau_2}$.
  We want to show that $\tempty{e_2}^{\hNetCtx,\hRandCtx,\hVarCtx[x\rightarrow \tempty{e_1}^{\hNetCtx,\hRandCtx,\hVarCtx}_L]}_L \in \ttype{\tau_2}$.
   By induction hypothesis we have that $\forall z. \tau \in \consNew{\id \rightarrow \tau_1}{\hcVarCtx}. (\hVarCtx[x\rightarrow \tempty{e_1}^{\hNetCtx,\hRandCtx,\hVarCtx}_L])[z] \in \ttype{\tau}
$ which follows from other induction hypothesis.

\textbf{Inductive Case 4.}
Suppose $e$ is a vector. By definition in Figure~\ref{fig:types} we have $\hcTypeRel{e_1}{\RealType}  \ldots  \hcTypeRel{e_n}{\RealType}$. This follows from induction hypothesis on each $ e_i $.

\end{proof}

\section{Proofs of Soundness of LDL Relative to \LJ}
\subsection{Proof of Lemma~\ref{lem:eval}}
\label{app:lemma}

We start with a lemma needed for base case of adequacy proofs for all FDLs.
 In the proofs, we reduce $\bowtie$ to $\{== , \leq ,\neq\}$, in the cases of other comparison operators can be proved analogously.
 
 \begin{lemma*}[Soundness of FDL comparisons] If  $\validForm{e_1 \bowtie e_2}$, then for all $L$ in fuzzy differentiable logics (FDL) the following holds:
 	
 	If $\tempty{e_1 \bowtie e_2}^{\hNetCtx,\hRandCtx, \hVarCtx}_{L} = \tempty{\top}^{\hNetCtx,\hRandCtx, \hVarCtx}_{L}$ then $\tempty{e_1}^{\hNetCtx,\hRandCtx, \hVarCtx}_{L} \bowtie \tempty{e_2}^{\hNetCtx,\hRandCtx, \hVarCtx}_{L}  = \top$. 
 	
 	If $\tempty{e_1 \bowtie e_2}^{\hNetCtx,\hRandCtx, \hVarCtx}_{L} = \tempty{\bot}^{\hNetCtx,\hRandCtx, \hVarCtx}_{L}$ then $\tempty{e_1}^{\hNetCtx,\hRandCtx, \hVarCtx}_{L} \bowtie \tempty{e_2}^{\hNetCtx,\hRandCtx, \hVarCtx}_{L}  = \bot$. 
 \end{lemma*} 
 

\begin{proof} 
The proof of Lemma~\ref{lem:eval} proceeds by case-reasoning on the comparison operators.
Since $e_1 \bowtie e_2$ is well-typed, each of $e_1, e_2$ has type $\RealType$. In that case, $\tempty{e_i}^{\semCtx} \in \Real$ by Lemma~\ref{lem:expr-real}. 

\textbf{Case 1.} If $\bowtie$ is $==$, then  $\tempty{e_1 == e_2}^{\semCtx}_{FDL} = 1$ means $\tempty{e_1}^{\semCtx}_{FDL} - \tempty{e_2}^{\semCtx}_{FDL} = 0$,
 that is, $\tempty{e_1}^{\semCtx}_{FDL} = \tempty{e_2}^{\semCtx}_{FDL}$. But then, $\tempty{e_1}^{\semCtx}_{FDL} ==^* \tempty{e_2}^{\semCtx}_{FDL} = \top$. Similarly, $\tempty{e_1 == e_2}_{FDL} = 0$ means $\tempty{e_1}^{\semCtx}_{FDL} - \tempty{e_2}^{\semCtx}_{FDL} \neq 0$, that is, $\tempty{e_1}^{\semCtx} \neq \tempty{e_2}^{\semCtx}_{FDL}$. But then, $\tempty{e_1}^{\semCtx}_{FDL} ==^* \tempty{e_2}^{\semCtx}_{FDL} = \bot$. 
	
	\textbf{Case 2.} If $\bowtie$ is $\leq$, then  $\tempty{e_1 \leq e_2}^{\semCtx}_{FDL} = 1$ means $\tempty{e_1}^{\hNetCtx,\hRandCtx, \hVarCtx}_{FDL} - \tempty{e_2}^{\semCtx}_{FDL} \leq 0$, that is, $\tempty{e_1}^{\semCtx}_{FDL} \leq \tempty{e_2}^{\semCtx}_{FDL}$. But then, $\tempty{e_1}^{\semCtx}_{FDL} \leq^* \tempty{e_2}^{\semCtx}_{FDL} = \top$. Similarly, $\tempty{e_1 \leq e_2}^{\semCtx}_{FDL} = 0$ means $\tempty{e_1}^{\semCtx}_{FDL} - \tempty{e_2}^{\semCtx}_{FDL} > 0$, that is, $\tempty{e_1}^{\semCtx}_{FDL} > \tempty{e_2}^{\semCtx}_{FDL}$. But then, $\tempty{e_1}^{\semCtx}_{FDL} \leq^* \tempty{e_2}^{\semCtx}_{FDL} = \bot$. 
	
	\textbf{Case 3.} If $\bowtie$ is $\neq$, then  $\tempty{e_1 \neq e_2}^{\semCtx}_{FDL} = 1$ means $[\tempty{e_1}^{\semCtx}_{FDL} = \tempty{e_2}^{\semCtx}_{FDL}] = 0$, that is, $\tempty{e_1}^{\semCtx}_{FDL} \neq \tempty{e_2}^{\semCtx}_{FDL}$. But then, $\tempty{e_1}^{\semCtx}_{FDL} \neq^* \tempty{e_2}^{\semCtx}_{FDL} = \top$. Similarly, $\tempty{e_1 \leq e_2}^{\semCtx}_{FDL} = 0$ means $[\tempty{e_1}^{\semCtx}_{FDL} = \tempty{e_2}^{\semCtx}_{FDL}] = 1$, that is, $\tempty{e_1}^{\semCtx}_{FDL} = \tempty{e_2}^{\semCtx}_{FDL}$. But then, $\tempty{e_1}^{\semCtx}_{FDL} \neq^* \tempty{e_2}^{\semCtx}_{FDL} = \bot$. 
	
	The remaining comparisons in $\bowtie$ are all defined using the already proven comparisons.
\end{proof}

%

\subsection{Proof of Soundness for G\"{o}del  DL}
\label{app:proof-godel}
The below proof is for Lemma~\ref{l:s1}.
\begin{proof}
	The mutually  inductive proof proceeds by case-reasoning on the shape of the formula $e$, and by induction on the structure of $\tGodel{}^{\semCtx}$. As finite quantifiers are defined via conjunctions, we only cover the case of infinite quantifiers explicitly. 
	As the proof is mutually inductive between two parts of Lemma~\ref{l:s1}, we will refer to them accordingly as 
	Lemma~\ref{l:s1} (Part 1) and Lemma~\ref{l:s1} (Part 2). 
	
	We start with Lemma~\ref{l:s1} (Part 1).
	The case of $e = \bot$ is automatically excluded, so our first base case is:

	\textbf{Base Case 1.} 
	Suppose $e = \top$. But we know $\SequentCLJ(){\top}$ by the rule $(\top)$.
	
	\textbf{Base Case 2.} 
	Suppose $e = e_1 \bowtie e_2$. As $\tGodel{e_1 \bowtie e_2}^{\semCtx} = 1$, 
	by Lemma~\ref{lem:eval}, we have $\tempty{e_1}^{\semCtx}_{FDL} \bowtie^* \tempty{e_2}^{\semCtx}_{FDL} = \top$. But then, we can derive $\SequentCLJ(){ e_1 \bowtie e_2}$ by the rules $ \LJArithR$ and $(\top)$.

	\textbf{Inductive Case 1.} Suppose $e = \neg e_1$, and therefore $\tGodel{e}^{\semCtx} = 1-\tGodel{e_1}^{\semCtx} = 1$. This is only possible when $\tGodel{e_1}^{\semCtx} = 0$. 
	Then, by Lemma~\ref{l:s1} (Part 2), we have
	$\SequentCLJ(e_1){}$. But then, by the rule $\LJNegR$, we can derive $\SequentCLJ(){\neg e_1}$.
	
	\textbf{Inductive Case 2.} Suppose $e =  e_1 \land e_2$, and therefore $\tGodel{e}^{\semCtx} = \min (\tGodel{e_1}^{\semCtx}, \tGodel{e_2}^{\semCtx}) = 1$. This means that  
	$\tGodel{e_1}^{\semCtx} = \tGodel{e_2}^{\semCtx} = 1 $. And, by the induction hypothesis, we have that $\SequentCLJ(){e_1}$ and $\SequentCLJ(){e_2}$. But then, by the rule $\LJConjR$, we have $\SequentCLJ(){e_1 \land e_2}$.
	
	\textbf{Inductive Case 3.} Suppose $e =  e_1 \impl e_2$, and therefore $\tGodel{e}^{\semCtx} = \max (1-\tGodel{e_1}^{\semCtx}, \tGodel{e_2}^{\semCtx}) = 1$. This means that at least one of the following is true: $1-\tGodel{e_1}^{\semCtx} = 1$ or $\tGodel{e_2}^{\semCtx} = 1$. 
	In the first case we have $\tGodel{e_1}^{\semCtx} = 0$. 
	Then, by Lemma~\ref{l:s1} (Part 2), we have
	$\SequentCLJ(e_1){}$. 
	In the second case, by the induction hypothesis, we have $\SequentCLJ(){e_2}$.
	In either case, using one of the weakening rules, we can obtain $\SequentCLJ(\Gamma_T, e_1){e_2}$. This allows us to use the rule $\LJImplR$ to derive $\SequentCLJ{e_1 \impl e_2}$.

	\textbf{Inductive Case 4.} Suppose $e =  \forall x: \tau. e_1$, and therefore we have $\tGodel{\forall x: \tau .\ e_1 }^{\semCtx} = 1$.  This means that the minimum expected value of $\lambda y. \tGodel{e_1}^{\hNetCtx,\hRandCtx, \hVarCtx[x\rightarrow y]}$ is $1$.  
	Seeing that $1$ is the top value, it means that $\tGodel{e_1}^{\hNetCtx,\hRandCtx, \hVarCtx[x\rightarrow y]} = 1$ for all inputs $y$. But then, by the induction hypothesis, we have $\SequentCLJ(){e_1}$. We therefore can deduce $\SequentCLJ(){\forall x. e_1}$ ($x \notin FV(\Gamma_T)$, as $\Gamma_T$ is empty). 
	
	\textbf{Inductive Case 5.} Suppose $e =  \exists x: \tau. e_1$,  and therefore we have $\tGodel{\exists x: \tau .\ e_1 }^{\hNetCtx,\hRandCtx, \hVarCtx} = 1$. This means that the maximum expected value of $\lambda y. \tGodel{e_1}^{\hNetCtx,\hRandCtx, \hVarCtx[x\rightarrow y]}$ is $1$. 
	Tt means that $\tGodel{e_1}^{\hNetCtx,\hRandCtx, \hVarCtx[x\rightarrow y]} = 1$ for at least one input $w$. But then, by the induction hypothesis, we have $\SequentCLJ(){e_1[y / x]}$. We therefore can deduce $\SequentCLJ(){\exists x. e_1}$.
	
	We now move on to the proof of Lemma~\ref{l:s1} (Part 2). The case of $e = \top$ is automatically excluded, so our first base case is:
	
	\textbf{Base Case 1.} 
	Suppose $e = \bot$. But we know $\SequentCLJ(\bot){}$ by the $(\bot)$.
	
	\textbf{Base Case 2.} 
	Suppose $e = e_1 \bowtie e_2$, where $\bowtie$ is a comparison operator. Moreover, $\tGodel {  e_1 \bowtie e_2}^{\semCtx} = 0$. 
	By Lemma~\ref{lem:eval}, we have $\tempty{e_1}^{\semCtx}_{FDL} \bowtie^* \tempty{e_2}^{\semCtx}_{FDL} = \bot$.
	But then we obtain a proof for $\SequentCLJ(\Gamma_T,e_1 \leq e_2){ }$ by the rules $ \LJArithL $ and ($\bot$).
	
	\textbf{Inductive Case 1.} 
	Suppose $e = \neg e_1$, and therefore $\tGodel{e}^{\semCtx} = 1-\tGodel{e_1}^{\semCtx} = 0$. This is only possible when $\tGodel{e_1}^{\semCtx} = 1$. 
	Then $\tGodel{e_1}^{\semCtx} \in \tGodel{\top}^{\semCtx}$. Then by Lemma~\ref{l:s1} (Part 1) we have $\SequentCLJ(){e_1 }$. But then by $\LJNegL$ we can derive $\SequentCLJ(\neg e_1){ }$.

	\textbf{Inductive Case 2.} 
	Suppose $e =  e_1 \land e_2$, and therefore $\tGodel{e}^{\semCtx} = \min (\tGodel{e_1}^{\semCtx}, \tGodel{e_2}^{\semCtx}) = 0$. This means that at least one of the following is true:
	$\tGodel{e_1}^{\semCtx} = 0$ or $ \tGodel{e_2}^{\semCtx} = 0 $. And, by the induction hypothesis, we have that $\SequentCLJ(e_1){}$ or $\SequentCLJ(e_2){}$. But then, by the rule $\LJConjL$, we have $\SequentCLJ(e_1 \land e_2){}$.
	
	\textbf{Inductive Case 3.} 
	Suppose $e =  e_1 \impl e_2$, and therefore ${\tGodel{e}^{\semCtx} = \max (1-\tGodel{e_1}^{\semCtx}, \tGodel{e_2}^{\semCtx}) = 0}$. This means that both $1-\tGodel{e_1}^{\semCtx} = 0$ and $\tGodel{e_2}^{\semCtx}  = 0$. 
	This gives us $\tGodel{e_1}^{\semCtx} = 1$, therefore by Lemma~\ref{l:s1} (Part 1) we obtain $\SequentCLJ(){e_1}$. For $\tGodel{e_2}^{\semCtx}  = 0$, we use the induction hypothesis and conclude $\SequentCLJ(e_2){}$.
	Using the rule $\LJImplL$, we obtain $\SequentCLJ(e_1 \impl e_2){}$.
	
	\textbf{Inductive Case 4.} Suppose $e =  \forall x. e_1$, and therefore we have $\tGodel{\forall x: \tau .\ e_1 }^{\hNetCtx,\hRandCtx, \hVarCtx} = 0$.  This means that the minimum expected value of $\lambda y. \tGodel{e_1}^{\hNetCtx,\hRandCtx, \hVarCtx[x\rightarrow y]}$ is $0$. 
	It means that $\tGodel{e_1}^{\hNetCtx,\hRandCtx, \hVarCtx[x\rightarrow y]} = 0$ for some input $w$. But then, by the induction hypothesis, we have $\SequentCLJ(e_1[y / w]){}$. We therefore can deduce $\SequentCLJ(\forall x. e_1){}$. 
	
	\textbf{Inductive Case 5.} Suppose $e =  \exists x. e_1$,  and therefore we have $\tGodel{\exists x: \tau .\ e_1 }^{\hNetCtx,\hRandCtx, \hVarCtx} = 0$. This means that the maximum expected value of $\lambda y. \tGodel{e_1}^{\hNetCtx,\hRandCtx, \hVarCtx[x\rightarrow y]}$ is $0$. 
	Seeing that $0$ is the bottom value, it means that $\tGodel{e_1}^{\hNetCtx,\hRandCtx, \hVarCtx[x\rightarrow y]} = 0$ for all inputs. But then, by the induction hypothesis, we have $\SequentCLJ(e_1){}$. We therefore can deduce $\SequentCLJ(\exists x. e_1){}$, ($x \notin FV(\Gamma_T)$, as $\Gamma_T$ is empty).
\end{proof}

\subsection{Proof of Soundness for Product DL}
We first define analogous theorem for another fuzzy logic - product, denoted $\tproduct{}$.

\begin{lemma}\label{th:p1}
	Given a formula $e$, for any contexts $N, \Gamma, Q$ the following hold:
	\begin{enumerate}
		\item if $\tproduct{e}^{\hNetCtx,\hRandCtx, \hVarCtx} = \tproduct{\top}^{\hNetCtx,\hRandCtx, \hVarCtx} $  then $\SequentCLJ(){e}$.
		\item if   
		$\tproduct{e}^{\hNetCtx,\hRandCtx, \hVarCtx} = \tproduct{\bot}^{\hNetCtx,\hRandCtx, \hVarCtx}$ 
		 then $\SequentCLJ(e){}$.  
	\end{enumerate}
\end{lemma}

\begin{proof}
	The mutually  inductive proof proceeds by case-reasoning on the shape of the formula $e$, and by induction on the structure of $\tproduct{}^{\semCtx}$, a choice of one of the DLs.
		As the proof is mutually inductive between two parts of Lemma~\ref{th:p1}, we will refer to them accordingly as 
	Lemma~\ref{th:p1} (Part 1) and Lemma~\ref{th:p1} (Part 2). 
	
	We start with the first part of Lemma~\ref{th:p1}.
	The case of $e = \bot$ is automatically excluded, so our first base case is:
	
	\textbf{Base Case 1.} 
	Suppose $e = \top$. But we know $\SequentCLJ(){\top}$ by the rule $(\top)$.
	
	\textbf{Base Case 2.} 
		Suppose $e = e_1 \bowtie e_2$. Moreover, $\tproduct {e_1 \bowtie e_2}^{\semCtx} = 1$. 
	By Lemma~\ref{lem:eval}, we have $\tempty{e_1}^{\semCtx}_{FDL} \bowtie^* \tempty{e_2}^{\semCtx}_{FDL} = \top$.
	But then, we can derive $\SequentCLJ(){ e_1 \bowtie e_2}$ by the rules $ \LJArithR$ and $(\top)$. 
	
	\textbf{Inductive Case 1.} Suppose $e = \neg e_1$, and therefore $\tproduct{e}^{\semCtx} = 1-\tproduct{e_1}^{\semCtx} = 1$. This is only possible when $\tproduct{e_1} = 0$. 
	Then, by Lemma~\ref{th:p1} (Part 2), we have
	$\SequentCLJ(e_1){}$. But then, by the rule $\LJNegR$, we can derive $\SequentCLJ(){\neg e_1}$.
	
	\textbf{Inductive Case 2.} Suppose $e =  e_1 \land e_2$, and therefore $\tproduct{e}^{\semCtx} = \tproduct{e_1}^{\semCtx} \times \tproduct{e_2}^{\semCtx} = 1$. Since $1$ is the top value that means we have $\tproduct{e_1}^{\semCtx} = \tproduct{e_2}^{\semCtx} = 1 $ And, by the induction hypothesis, we have that $\SequentCLJ(){e_1}$ and $\SequentCLJ(){e_2}$. But then, by the rule $\LJConjR$, we have $\SequentCLJ(){e_1 \land e_2}$.
	
	\textbf{Inductive Case 3.} Suppose $e =  e_1 \impl e_2$, and therefore $\tproduct{e}^{\semCtx} = 1-\tproduct{e_1}^{\semCtx}+ \tproduct{e_1}^{\semCtx} \times \tproduct{e_2}^{\semCtx} = 1$.
	This means that $\tproduct{e_1}^{\semCtx}= \tproduct{e_1}^{\semCtx} \times \tproduct{e_2}^{\semCtx}$. From this at least one of the following is true $\tproduct{e_1}^{\semCtx} = 0$ or $\tproduct{e_2}^{\semCtx} = 1 $.
	In the first case by Lemma~\ref{th:p1} (Part 2) we have $\SequentCLJ(e_1){}$. In the second case, by induction hypothesis we have $\SequentCLJ(){e_2}$.
	Now using one of the weakening rules, we can obtain $\SequentCLJ(e_1){e_2}$. This allows us to use the rule $\LJImplR$ to derive $\SequentCLJ(){e_1 \impl e_2}$.

	\textbf{Inductive Case 4.} Suppose $e =  \forall x. e_1$, and therefore we have $\tproduct{\forall x: \tau .\ e_1 }^{\hNetCtx,\hRandCtx, \hVarCtx} = 1$.  This means that the minimum expected value of $\lambda y. \tproduct{e_1}^{\hNetCtx,\hRandCtx, \hVarCtx[x\rightarrow y]}$ is $1$.   
	Seeing that $1$ is the top value, it means that $\tproduct{e_1}^{\hNetCtx,\hRandCtx, \hVarCtx[x\rightarrow y]} = 1$ for all inputs $y$. But then, by the induction hypothesis, we have $\SequentCLJ(){e_1}$. We therefore can deduce $\SequentCLJ(){\forall x. e_1}$ ($x \notin FV(\Gamma_T)$, as $\Gamma_T$ is empty). 
	
	\textbf{Inductive Case 5.} Suppose $e =  \exists x. e_1$,  and therefore we have $\tproduct{\exists x: \tau .\ e_1 }^{\hNetCtx,\hRandCtx, \hVarCtx} = 1$. This means that the maximum expected value of $\lambda y. \tproduct{e_1}^{\hNetCtx,\hRandCtx, \hVarCtx[x\rightarrow y]}$ is $1$. 
	It means that $\tproduct{e_1}^{\hNetCtx,\hRandCtx, \hVarCtx[x\rightarrow y]} = 1$ for at least one input $w$. But then, by the induction hypothesis, we have $\SequentCLJ(){e_1[y / x]}$. We therefore can deduce $\SequentCLJ(){\exists x. e_1}$.
	
	We now move on to the proof of the second part of Lemma~\ref{th:p1}. The case of $e = \top$ is automatically excluded, so our first base case is:
	
	\textbf{Base Case 1.} 
	Suppose $e = \bot$. But we know $\SequentCLJ(\bot){}$ by the $(\bot)$.
	
	\textbf{Base Case 2.} 
		Suppose $e = e_1 \bowtie e_2$, where $\bowtie$ is a comparison operator, and $e_1, e_2$ are real numbers. Moreover, $\tproduct {  e_1 \bowtie e_2}^{\semCtx} = 0$. 
	By Lemma~\ref{lem:eval}, we have $\tempty{e_1}^{\semCtx}_{FDL} \bowtie^* \tempty{e_2}^{\semCtx}_{FDL} = \bot$.
	But then we obtain a proof for $\SequentCLJ(\Gamma_T,e_1 \leq e_2){ }$ by the rules $ \LJArithL $ and ($\bot$).
	
	\textbf{Inductive Case 1.} 
	Suppose $e = \neg e_1$, and therefore $\tproduct{e}^{\semCtx} = 1-\tproduct{e_1}^{\semCtx} = 0$. This is only possible when $\tproduct{e_1}^{\semCtx} = 1$. 
	Then $\tproduct{e_1}^{\semCtx} \in \tproduct{\top}^{\semCtx}$. Then by Lemma~\ref{th:p1} (Part 1) we have $\SequentCLJ(){e_1 }$. But then by $\LJNegL$ we can derive $\SequentCLJ(\neg e_1){ }$.

	\textbf{Inductive Case 2.} 
	Suppose $e =  e_1 \land e_2$, and therefore $\tproduct{e}^{\semCtx} = \tproduct{e_1}^{\semCtx} \times \tproduct{e_2}^{\semCtx} = 0$.
	This means that at east one of the following holds - $ \tproduct{e_1}^{\semCtx} = 0 $ or $ \tproduct{e_2}^{\semCtx} = 0 $. In both cases respectively by induction hypothesis we have $\SequentCLJ(e_1){}$ or $\SequentCLJ(e_2){}$. But then, by the rule $(\LJConjL)$, we have $\SequentCLJ(e_1 \land e_2){}$.
	
	\textbf{Inductive Case 3.} 
	Suppose $e =  e_1 \impl e_2$, and therefore $\tproduct{e}^{\semCtx} = 1-\tproduct{e_1}^{\semCtx}+ \tproduct{e_1}^{\semCtx} \times \tproduct{e_2}^{\semCtx} = 0$. 
	This means that $\tproduct{e_1}^{\semCtx}-1= \tproduct{e_1}^{\semCtx} \times \tproduct{e_2}^{\semCtx}$. However as $\tproduct{e_1}^{\semCtx} \times \tproduct{e_2}^{\semCtx} \in [0,1]$ that means that $\tproduct{e_1}^{\semCtx}-1 \in [0,1]$ and therefore $\tproduct{e_1}^{\semCtx} = 1$ and by Lemma~\ref{th:p1} (Part 1) we obtain $\SequentCLJ(){e_1}$. From this we have $\tproduct{e_2}^{\semCtx} = 0$ and we use the induction hypothesis and conclude $\SequentCLJ(e_2){}$.
	Using the rule $\LJImplL$, we obtain $\SequentCLJ(e_1 \impl e_2){}$.

	\textbf{Inductive Case 4.} Suppose $e =  \forall x. e_1$, and therefore we have $\tproduct{\forall x: \tau .\ e_1 }^{\hNetCtx,\hRandCtx, \hVarCtx} = 0$.  This means that the minimum expected value of $\lambda y. \tproduct{e_1}^{\hNetCtx,\hRandCtx, \hVarCtx[x\rightarrow y]}$ is $0$. 
	It means that $\tproduct{e_1}^{\hNetCtx,\hRandCtx, \hVarCtx[x\rightarrow y]} = 0$ for some input $w$. But then, by the induction hypothesis, we have $\SequentCLJ(e_1[y / x]){}$. We therefore can deduce $\SequentCLJ(\forall x. e_1){}$. 
	
	\textbf{Inductive Case 5.} Suppose $e =  \exists x. e_1$,  and therefore we have $\tproduct{\exists x: \tau .\ e_1 }^{\hNetCtx,\hRandCtx, \hVarCtx} = 0$. This means that the maximum expected value of $\lambda y. \tproduct{e_1}^{\hNetCtx,\hRandCtx, \hVarCtx[x\rightarrow y]}$ is $0$. 
	Seeing that $0$ is the bottom value, it means that $\tproduct{e_1}^{\hNetCtx,\hRandCtx, \hVarCtx[x\rightarrow y]} = 0$ for all inputs. But then, by the induction hypothesis, we have $\SequentCLJ(e_1){}$. We therefore can deduce $\SequentCLJ(\exists x. e_1){}$, ($x \notin FV(\Gamma_T)$, as $\Gamma_T$ is empty).
\end{proof}

Its corollary is:

\begin{theorem}[Soundness of Product DL]
	Given a formula $e$, for any contexts $N, \Gamma, Q$ if$\tproduct{e}^{\hNetCtx,\hRandCtx, \hVarCtx} = \tproduct{\top}^{\hNetCtx,\hRandCtx, \hVarCtx} = 1$ then $\SequentCLJ(){e}$.
\end{theorem}

\subsection{Proof of Soundness for DL2}
\label{app:dl2-adequacy}

We now provide the proof for Theorem~\ref{th:d1}. It is important to remember that DL2 translation does not include a stand-alone negation operator or implication - therefore this proof is done for a limited version of \LJ.	We start with a helper lemma.

\begin{lemma}[Adequacy of intervals and arithmetic operations in DL2]\label{lem:eval1} 
If  $\validForm{e_1 \bowtie e_2} : \RealType$ and  $\tDLtwo{e_1 \bowtie e_2}^{\semCtx} = 0$ then $\tempty{e_1}^{\semCtx}_{DL2} \bowtie \tempty{e_2}^{\semCtx}_{DL2} = \top$. 

\end{lemma}
\begin{proof} 
 Since $e_1 \bowtie e_2$ is well-typed each of $e_1, e_2$ has type $\RealType$.
  If $\bowtie$ is $==$, then  $\tDLtwo{e_1 == e_2}^{\semCtx} = 0$ means $\tDLtwo{e_1}^{\semCtx} - \tDLtwo{e_2}^{\semCtx} = 0$, that is, $\tDLtwo{e_1}^{\semCtx} = \tDLtwo{e_2}^{\semCtx}$. But then, $\tempty{e_1}^{\semCtx}_{DL2}  == \tempty{e_2}^{\semCtx}_{DL2}  = \top$. 
	
	If $\bowtie$ is $\leq$, then  $\tDLtwo{e_1 \bowtie e_2}^{\semCtx} = 0$ means $\tDLtwo{e_1}^{\semCtx} - \tDLtwo{e_2}^{\semCtx} \leq 0$, that is, $\tDLtwo{e_1}^{\semCtx} \leq \tDLtwo{e_2}^{\semCtx}$. But then, $\tempty{e_1}^{\semCtx}_{DL2}  \leq \tempty{e_2}^{\semCtx}_{DL2}  = \top$.
	
	If $\bowtie$ is $\neq$, then  $\tDLtwo{e_1 \neq e_2}^{\semCtx} = 0$ means $[\tDLtwo{e_1}^{\semCtx} == \tDLtwo{e_2}^{\semCtx}] = 0$, that is, $\tDLtwo{e_1}^{\semCtx} \neq \tDLtwo{e_2}^{\semCtx}$. But then, $\tempty{e_1}^{\semCtx}_{DL2}  \neq \tempty{e_2}^{\semCtx}_{DL2}  = \top$.  
	
	The remaining comparisons in $\bowtie$ are all defined using the already proven comparisons.
\end{proof}

We now move on to the proof of Theorem~\ref{th:d1}.

\begin{theorem*}[Soundness of DL2]
	Given a formula $e$, taking DL2 with just connectives $\land, \lor$ and quantifiers, for any contexts $N, \Gamma, Q$:
	if $\tDLtwo{e}^{\hNetCtx,\hRandCtx, \hVarCtx} = \tDLtwo{\top}^{\hNetCtx,\hRandCtx, \hVarCtx}$  then $\SequentCLJ(){e}$. 
\end{theorem*}

\begin{proof}
	\textbf{Base Case 1.} 
	Suppose $e = \top$. But we know $\SequentCLJ(){\top}$ by the rule $(\top)$.
	
	\textbf{Base Case 2.} 
	Suppose $e = e_1 \bowtie e_2$, where $\bowtie$ is a comparison operator, and $e_1, e_2$ are real numbers. Moreover, $\tDLtwo {  e_1 \bowtie e_2}^{\semCtx} = 0$. 
	By Lemma~\ref{lem:eval1}, we have $\tempty{e_1}^{\hNetCtx,\hRandCtx, \hVarCtx}_{DL2} \bowtie^* \tempty{e_2}^{\hNetCtx,\hRandCtx, \hVarCtx}_{DL2} = \top$.
	But then we obtain a proof for $\SequentCLJ{e_1 \bowtie e_2}$ by the rules $ \LJArithR $ and ($\top$).
	
	\textbf{Inductive Case 1.} Suppose $e =  e_1 \land e_2$, and therefore ${tDLtwo{e}^{\semCtx} = \tDLtwo{e_1}^{\semCtx} + \tDLtwo{e_2}^{\semCtx} = 0}$.That means we have $\tDLtwo{e_1}^{\semCtx} = \tDLtwo{e_2}^{\semCtx} = 0 $ And, by the induction hypothesis, we have that $\SequentCLJ(){e_1}$ and $\SequentCLJ(){e_2}$. But then, by the rule $\LJConjR$, we have $\SequentCLJ(){e_1 \land e_2}$.
	
	\textbf{Inductive Case 2.} Suppose $e =  \forall x. e_1$, and therefore we have $\tDLtwo{\forall x: \tau .\ e_1 }^{\hNetCtx,\hRandCtx, \hVarCtx} = 0$.  This means that the minimum expected value of $\lambda y. \tDLtwo{e_1}^{\hNetCtx,\hRandCtx, \hVarCtx[x\rightarrow y]}$ is $0$.  
	Seeing that $0$ is the top value, it means that $\tDLtwo{e_1}^{\hNetCtx,\hRandCtx, \hVarCtx[x\rightarrow y]} = 0$ for all inputs $y$. But then, by the induction hypothesis, we have $\SequentCLJ(){e_1}$. We therefore can deduce $\SequentCLJ(){\forall x. e_1}$ ($x \notin FV(\Gamma_T)$, as $\Gamma_T$ is empty). 
	
	\textbf{Inductive Case 3.} Suppose $e =  \exists x. e_1$,  and therefore we have $\tDLtwo{\exists x: \tau .\ e_1 }^{\hNetCtx,\hRandCtx, \hVarCtx} = 0$. This means that the maximum expected value of $\lambda y. \tDLtwo{e_1}^{\hNetCtx,\hRandCtx, \hVarCtx[x\rightarrow y]}$ is $0$. 
	It means that $\tDLtwo{e_1}^{\hNetCtx,\hRandCtx, \hVarCtx[x\rightarrow y]} = 0$ for at least one input $w$. But then, by the induction hypothesis, we have $\SequentCLJ(){e_1[y / x]}$. We therefore can deduce $\SequentCLJ(){\exists x. e_1}$.
\end{proof}

\section{Further Discussion of Logical and Geometric Properties}
\label{app:properties}

	The DLs defined as part of LDL have the following properties:
	\begin{itemize}
		\item $\tDLtwo{}$ is commutative, scale-invariant, associative, sound (for a limited \LJ) and has shadow-lifting. It is not idempotent and does not have quantifier commutativity. While its semantics for logical connectives is weakly smooth, $\tDLtwo{}$ is not weakly smooth due to the presence of $\min$ and $\max$ in translation of comparisons. While no proofs about the properties were provided in~\citet{fischer2019dl} the translation of conjunction in DL2 and product based fuzzy logic is standard addition. The more common properties including idempotence, commutativity, associativity and min-max boundedness in given domain are known and proven. Considering the semantics of conjunction is addition it is simple enough to reason about its partial derivatives and therefore shadow-lifting. 
		\item $\tGodel{}$ is idempotent, commutative, scale-invariant, associative, sound and has quantifier commutativity. It does not have shadow-lifting and it is not weakly smooth.   The semantics of conjunction is a minimum between the two elements which prohibits shadow-lifting as well as smoothness by definition. The simpler properties are well known and proved in fuzzy logic literature as it is an established fuzzy logic.
		\item $\tlukasiewicz{}$ is  commutative, sound (for a syntax excluding either negation or implication) and associative. It is not idempotent, scale-invariant, weakly smooth and does not have shadow-lifting or quantifier commutativity.  The presence of maxima in the semantics of conjunction however naturally prohibits shadow lifting and smoothness. Similarly to Gödel it is an established fuzzy logic and many of the properties of its semantics of conjunction have been proven.
		\item $\tyager{}$ is  commutative, sound (for a limited syntax) and associative. It is not idempotent, scale-invariant, weakly smooth and does not have shadow-lifting or quantifier commutativity.  Similarly to the previous two the presence of maxima prohibits shadow lifting and smoothness. This is another one of the well established fuzzy logics for which majority of the properties have been already investigated by the community.
		\item $\tproduct{}$ is  commutative, associative, sound and has shadow-lifting. It is not idempotent,  scale-invariant, does not have quantifier commutativity.  While its semantics of logical connectives are weakly smooth $\tproduct{}$ is not weakly smooth due to the semantics of comparisons.  Furthermore as its semantics of conjunction is the arithmetic operation of multiplication it is easy to reason about its partial derivates and therefore shadow lifting. Since it is a well-established fuzzy logic its properties the more common logical properties have been extensively studied already.
		\item $\tSTL{}$ is  idempotent, commutative, scale-invariant, weakly smooth and has shadow-lifting. It is not associative or sound and does not have quantifier commutativity. The properties of this translation have been proven in the paper it originates from as indicated in Table~\ref{tab:properties}, aside from quantifier commutativity which has been added in this paper. The translation has been changed for the purposes of LDL however the proofs would remain analogous. Furthermore it is the only DL for which the soundness does not hold, for either full or limited \LJ. 

\item Only $\tGodel{}$'s quantifiers commute with connectives, as its connectives are defined via $max$ and $min$ (commuting with maxima and minima used in quantifier interpretation). 

	\end{itemize}


%
%

\end{document}